\documentclass{article}

\usepackage{fullpage}
\usepackage{hyperref}
\usepackage{mathtools}
\usepackage{amssymb}
\usepackage{amsthm}
\usepackage{tikz}
\usepackage{graphicx}
\usepackage{paralist}
\usepackage[basic,classfont=caps]{complexity}
\usepackage[T1]{fontenc}
\usepackage{enumitem}
\usepackage{xspace}
\usepackage{todonotes}
\usepackage{stmaryrd}
\usepackage{microtype}

\makeatletter
\let\epsilon\varepsilon
\let\phi\varphi
\let\emptyset\varnothing
\let\rho\varrho
\makeatother

\newtheorem{theorem}{Theorem}
\newtheorem{lemma}{Lemma}
\newtheorem{proposition}[lemma]{Proposition}
\newtheorem{corollary}[lemma]{Corollary}

\theoremstyle{definition}
\newtheorem{definition}{Definition}

\theoremstyle{remark}
\newtheorem{remark}{Remark}

\usetikzlibrary{positioning,fit,arrows,automata,calc,shapes,decorations.pathmorphing}
\tikzset{->,>=stealth',
shorten >=1pt,shorten <=1pt,
auto,node distance=1.5cm,
every loop/.style={looseness=6},
initial text={},
every state/.style={inner sep=0.2mm, minimum size=0.5cm},
el/.style={font=\scriptsize},
every fit/.style={draw,densely dotted,rectangle,inner sep=2mm},
loopright/.style={loop,looseness=6,out=30, in=-30},
loopleft/.style={loop,looseness=6,out=210, in=150},
loopabove/.style={loop,looseness=6,out=120, in=60},
loopbelow/.style={loop,looseness=6,out=300, in=240},
}

\renewclass{\EXP}{EXPtime}
\renewclass{\NEXP}{NEXPtime}
\renewclass{\coNEXP}{coNEXPtime}
\renewclass{\EXPSPACE}{EXPspace}
\renewclass{\AP}{APtime}
\renewclass{\PSPACE}{Pspace}

\newcommand{\Obs}{\mathsf{Obs}}
\newcommand{\calB}{\mathcal{B}}
\newcommand{\calF}{\mathcal{F}}
\newcommand{\calI}{\mathcal{I}}
\newcommand{\tadam}{\mathcal{T}_\forall}
\newcommand{\teve}{\mathcal{T}_\exists}
\newcommand{\supp}{\mathsf{supp}}
\newcommand{\post}{\mathsf{post}}
\newcommand{\prop}{\mathsf{prop}}
\newcommand{\pprop}{\mathsf{pw\text{-}prop_\calI}}
\newcommand{\suc}[1]{\textsf{succ}_\mathcal{I}(#1)}
\newcommand{\Plays}{\mathsf{Plays}}
\newcommand{\Prefs}{\mathsf{Prefs}}
\newcommand{\uMP}{\underline{\mathsf{MP}}}
\newcommand{\oMP}{\overline{\mathsf{MP}}}
\newcommand{\st}{\mathrel{\mid}}
\newcommand{\ie}{\textit{i.e.}\xspace}
\newcommand{\eg}{\textit{e.g.}\xspace}
\renewcommand{\restriction}[2]{{#1}\mathord{\downharpoonright}_{#2}}
\newcommand{\reset}{\mathsf{reset}}
\newcommand{\out}{\mathsf{Out}}

\begin{document}

% header
\title{Mean-payoff Games with Partial Observation\thanks{Work
partially supported by ERC Starting grant inVEST (FP7-279499).}}
\author{Paul Hunter, Arno Pauly, Guillermo A. P\'erez\thanks{F.R.S.-FNRS
	Aspirant and FWA post-doc fellow},
	and Jean-Fran\c{c}ois Raskin\\
D\'{e}partament d'Informatique, Universit\'{e} libre de Bruxelles\\
\texttt{\{{paul.hunter},
	{apauly},
	{guillermo.perez},
	{jraskin}\}
	{@ulb.ac.be}}}

\maketitle

\begin{abstract}
Mean-payoff games are important quantitative models for open reactive systems.
They have been widely studied as games of full observation. In this paper we
investigate the algorithmic properties of several sub-classes of mean-payoff
games where the players have asymmetric information about the state of the game.
These games are in general undecidable and not determined according to the
classical definition. We show that such games are determined under a more
general notion of winning strategy. We also consider mean-payoff games where the
winner can be determined by the winner of a finite cycle-forming game. This
yields several decidable classes of mean-payoff games of asymmetric information
that require only finite-memory strategies, including a generalization of
full-observation games where positional strategies are sufficient. We give an
exponential time algorithm for determining the winner of the latter.
\end{abstract}

\section{Introduction}
Mean-payoff games (MPGs) are two-player, infinite duration, turn-based games
played on finite edge-weighted graphs. The two players alternately move a token
around the graph; and one of the players (Eve) tries to maximize the (limit)
average weight of the edges traversed, whilst the other player (Adam) attempts
to minimize the average weight. Such games are particularly useful in the field
of verification of models of reactive systems, and the full-observation
versions of these games have been extensively
studied~\cite{em79,bsv04,cdh08,cdhr10}. One of the major open questions in the
field of verification is whether the following decision problem, known to be in
the intersection of the classes \NP~and \coNP~\cite{em79}\footnote{From results
in~\cite{zp96} and~\cite{jurdzinski98} it follows that the problem is
also in $\UP \cap \coUP$.}, can be solved in polynomial time: Given a threshold
$\nu$, does Eve have a strategy to ensure a mean-payoff value of at least $\nu$?

In game theory the concepts of partial and limited observation
indicate situations where players are uncertain about the state of
the game. In the context of verification games this partial knowledge is
reflected in one or both players being unable to determine the precise location
of the token amongst several equivalent states, and such games have also been
extensively studied~\cite{reif84,kv00,bd08,bcdhr08,ddgrt10}.
Adding partial observation to verification games results in an enormous increase
in complexity, both algorithmically and in terms of strategy synthesis.  For
example, it was shown in~\cite{ddgrt10} that for MPGs with partial observation,
when the mean payoff value is defined using $\limsup$, the analogue of the above
decision problem (\ie the threshold problem) is undecidable; and whilst
positional strategies suffice for MPGs with full observation, infinite memory
may be required. The first result of this paper is to show that this is also the
case when the mean payoff value is defined using the $\liminf$
operator, closing two open questions posed in~\cite{ddgrt10}.

These unfavourable results motivate the main investigation of this paper:
identifying classes of MPGs with partial observation where determining the
winner is decidable and where strategies with finite memory, possibly
positional, are sufficient.

To simplify our definitions and algorithmic results we initially consider a
restriction on the set of observations which we term \emph{limited observation}.
In games of limited observation the current observation contains only those
states consistent with the observable history, that is the observations are the
\emph{belief set of Eve} (see, \eg~\cite{cd12}). This is not too restrictive as
any MPG with partial observation can be realized as a game of limited
observation via a subset construction.  In Section~\ref{sec:pure-imp} we
consider the extension of our definitions to MPGs with partial observation via
this construction.

Our focus for the paper will be on games at the observation level, in particular
we are interested in \emph{observation-based strategies} for both players.
Whilst observation-based strategies for Eve are usual in the literature,
observation-based strategies for Adam have not, to the best of our knowledge,
been considered.  Such strategies are more advantageous for Adam as they
encompass several simultaneous concrete strategies.  Further, in games of
limited observation there is guaranteed to be at least one concrete strategy
consistent with an observation-based strategy.  Our second result is to show
that, although MPGs with partial observation are not determined under the usual
definition of (concrete) strategy, they are determined when Adam can use an
observation-based strategy.

In full-observation games, one aspect of MPGs that leads to simple (but not
quite efficient) decision procedures is their equivalence to finite
cycle-forming games. Such games are played as their infinite counterparts, however
when the token revisits a state the game is stopped. The winner is determined
by a finite analogue of the mean-payoff condition on the cycle now formed; that
is, Eve wins if the average weight of the edges traversed in the cycle exceeds a
given threshold. Ehrenfeucht and Mycielski~\cite{em79} and Bj\"orklund et
al.~\cite{bsv04}\footnote{A recent result of Aminof and Rubin~\cite{ar14}
corrects some errors in~\cite{bsv04}.} used this equivalence to show that
positional strategies are sufficient to win MPGs with full observation and
this leads to an $\NP \cap \coNP$ procedure for determining the winner.
Critically, a winning strategy in the finite game translates directly to a
winning strategy in the MPG, so such games are especially useful for strategy
synthesis.

We extend this idea to games of partial observation by introducing a finite,
full-observation, cycle-forming game played at the observation level.  That
is, the game finishes when an observation is revisited (though not necessarily
the first time).  In this reachability game winning strategies can be translated
to finite-memory winning strategies in the MPG. This leads to a large, natural
subclass of MPGs with partial observation, \emph{forcibly terminating} games,
where determining the winner is decidable and finite-memory observation-based
strategies suffice.

Unfortunately, recognizing if an MPG is a member of this class is undecidable,
and although determining the winner is decidable, we show that this problem is
complete (under polynomial-time reductions) for the class of all decidable
problems. Motivated by these negative algorithmic results, we investigate two
natural refinements of this class for which winner determination and class
membership are decidable. The first, \emph{forcibly first abstract cycle} games
(forcibly FAC games, for short), is the natural class of games obtained when our
cycle-forming game is restricted to simple cycles. Unlike the full-observation
case, we show that winning strategies in this finite simple
cycle-forming game may still require memory, though this memory is at most
exponential in the size of the game. The second refinement, \emph{first abstract
cycle} (FAC) games, is a further structural refinement that guarantees a winner
in the simple cycle-forming game. We show that in this class of games
positional observation-based strategies suffice.

\begin{table}[t]
\begin{center}
	\begin{tabular}{ |l|l|l|l|l| }
		\hline
		\multicolumn{2}{|l|}{} & Sufficient & Class & Winner \\
		\multicolumn{2}{|l|}{} & memory & membership & determination \\
		\hline
		\hline
		\multicolumn{2}{|l|}{Forcibly terminating} & Finite
		& Undecidable & \R-c \\
		\multicolumn{2}{|l|}{} & (Thm.~\ref{thm:fin-mem-det})
		& (Thm.~\ref{thm:undec-ft}) & (Thm.~\ref{thm:r-complete}) \\
		\hline
		Forc. & limited obs. & Exponential & \PSPACE-c & \PSPACE-c \\
		FAC & & (Thm.~\ref{thm:expMemory}) & (Thm.~\ref{thm:isadeqpure})
		& (Thm.~\ref{thm:apwd}) \\
		\cline{2-5}
		&  partial obs. & Doubly & \NEXP-h,
		& \EXP-c \\
		& & exponential & in \EXPSPACE & (Thm.~\ref{thm:pureWin}) \\
		& & (Thms.~\ref{thm:expMemory},~\ref{thm:exp-mem-po}) & (Thm.~\ref{thm:isimppure}) & \\
		\hline
		FAC & limited obs. & Positional & \coNP-c & $\NP \cap \coNP$ \\
		& & (Thm.~\ref{thm:posDet}) & (Thm.~\ref{thm:pureCM})
		& (Thm.~\ref{thm:windet-fac}) \\
		\cline{2-5}
		& partial obs. & Exponential & \coNEXP-c & \EXP-c \\
		&  & (Thms.~\ref{thm:posDet},~\ref{thm:exp-mem-po}) &
		(Thm.~\ref{thm:isimppure}) & (Thm.~\ref{thm:pureWin})\\
		\hline
	\end{tabular}
\caption{Summary of results for the classes of games studied.}
\end{center}
\label{table:results}
\end{table}
The sub-classes of MPGs with limited observation we study
then give rise to sub-classes of MPGs with partial observation. For the class
membership problem we show there is, as expected, an exponential blow-up in the
complexity, however for the problem of determining the winner the algorithmic
cost is significantly better.

Table~\ref{table:results} summarizes the results of this paper. An extended abstract
of this work appeared in~\cite{hpr14}.

\section{Preliminaries}
\paragraph*{Mean-payoff games}
A \emph{mean-payoff game (MPG) with partial observation} is a tuple $G = (
Q, q_I, \Sigma, \Delta, w, \Obs )$, where $Q$ is a finite set of states,
$q_I \in Q$ is the initial state, $\Sigma$ is a finite set of action symbols, $\Delta
\subseteq Q \times \Sigma \times Q$ is the transition relation, $w:\Delta \to
\mathbb{Z}$ is the weight function, and $\Obs \subset 2^Q$ is a partition of $Q$
into observations. We assume $\Delta$ is total, that is,
for every $(q,\sigma) \in Q \times \Sigma$ there exists $q' \in Q$ such that
$(q,\sigma,q') \in \Delta$.  We say that $G$ is a \emph{mean-payoff game with
limited observation} if additionally, $\Obs$ satisfies the following:
\begin{enumerate}[nolistsep,label=(\arabic*)]
	\item $\{q_I\} \in \Obs$, and
	\item For each $(o,\sigma) \in \Obs \times \Sigma$ the set $\{q' \in Q \st
		\exists q \in o \textrm{ and }(q,\sigma,q') \in
		\Delta\}$ is a union of elements of $\Obs$.
\end{enumerate}
Note that condition (2) is equivalent to saying that if $q \in o$, $q' \in o'$
and $(q,\sigma,q') \in \Delta$ then for every $r' \in o'$ there exists $r \in o$
such that $(r,\sigma,r') \in \Delta$. If every element of $\Obs$ is a singleton,
then we say $G$ is a \emph{mean-payoff game with full observation}. For
simplicity, we denote by $\post_\sigma(s) = \{q' \in Q \st \exists q \in s :
(q, \sigma, q') \in \Delta \}$ the set of $\sigma$-successors of a set of states
$s \subseteq Q$.

Figure~\ref{fig:purendet} gives an example of an MPG with limited observation,
with $\Sigma = \{a,b\}$ and $\Obs = \{\{q_0\},\{q_1,q_2\},\{q_3\}\}$. In this
work, unless explicitly stated otherwise, we depict states from an MPG with
partial observation as circles and transitions as arrows labelled by an
action-weight pair: $\sigma,w$. Observations are represented by dashed boxes.

\paragraph*{Abstract \& concrete paths} A \emph{concrete path} in an MPG with
partial observation is a sequence $q_0 \sigma_0 q_1 \sigma_1 \dots$ where for
all $i \ge 0$ we have $q_i \in Q$, $\sigma_i \in \Sigma$ and
$(q_i,\sigma_i,q_{i+1}) \in \Delta$. An \emph{abstract path} is a sequence $o_0
\sigma_0 o_1 \sigma_1 \dots$ where $o_i \in \Obs$, $\sigma_i \in \Sigma$ and for
all $i \ge 0$ there exists $q_i \in o_i$ and $q_{i+1} \in o_{i+1}$ with $(q_i,
\sigma_i, q_{i+1}) \in \Delta$. Given an abstract path $\psi$, let
$\gamma(\psi)$ be the (possibly empty) set of concrete paths that agree with the
observation and action sequence. In other words $\gamma(\psi) = \{ q_0 \sigma_0
q_1 \sigma_1 \dots \st \forall i \ge 0 : q_i \in o_i \text{ and }
(q_i, \sigma, q_{i+1}) \in \Delta\}$. Note that in games of limited observation
this set is never empty. Also, given an abstract (respectively
concrete) path $\psi$, let $\psi[..n]$ represent the prefix of $\psi$ up
to the $(n+1)$-th observation (state), which we express as
$\psi[n]$; similarly, we denote by $\psi[\ell..]$ the suffix of
$\psi$ starting from the $(\ell+1)$-th observation (state) and by $\psi[\ell..n]$
the finite sub-sequence starting and ending in the aforementioned
locations.

\paragraph*{Cycles}
An \emph{abstract (respectively concrete) cycle} is an abstract
(concrete) path $\chi = o_0 \sigma_0 \dots o_n$ where $o_0 = o_n$.
We say $\chi$ is \emph{simple} if $o_j \neq o_i$ for $0 \leq i < j <n$.
Given $k \in \mathbb{N}$ define $\chi^k$ to be the abstract (concrete) cycle
obtained by traversing $k$ times $\chi$. That is, $\chi^k = o_0' \sigma_0'
\dots o_{nk}'$ where for all $0 \le j \le nk$ we have that
$o_j' = o_{j \pmod{n}}$ and $\sigma_j' = \sigma_{j \pmod{n}}$.  A \emph{cyclic
permutation} of $\chi$ is an abstract (concrete) cycle $o_0' \sigma_0' \dots
o_n'$ such that $o_j' = o_{j+k \pmod{n}}$ and $\sigma_j' = \sigma_{j+k
\pmod{n}}$ for some $k \in \mathbb{N}$.  If $\chi' = o_0' \sigma_0' \dots o_m'$ is a
cycle with $o_0' = o_i$ for some $0 \le i < n$, the \emph{interleaving} of $\chi$
and $\chi'$ is the cycle $o_0 \sigma_0 \dots o_i \sigma_0' \dots o_m'
\sigma_i \dots o_n$.

\paragraph*{The mean payoff}
Given an infinite concrete path $\pi = q_0\sigma_0 q_1 \sigma_1 \dots$, the
\emph{payoff} up to the $(n+1)$-th element is given by
\[
	w(\pi[..n]) = \sum\limits_{i=0}^{n-1} w(q_i, \sigma_i, q_{i+1}).
\]
If $\pi$ is infinite, we define two \emph{mean payoff} values $\uMP$
and $\oMP$ as:

\[
	\uMP(\pi) = \liminf\limits_{n \rightarrow \infty} \frac{1}{n}
	w(\pi[..n])
	\qquad
	\oMP(\pi) = \limsup\limits_{n \rightarrow \infty} \frac{1}{n}
	w(\pi[..n])
\]

\paragraph*{Plays \& strategies} A play in an MPG with partial observation
$G$ is an infinite abstract path starting at $o_I \in \Obs$ where $q_I \in o_I$.
Denote by $\Plays(G)$ the set of all plays and by $\Prefs(G)$ the set of all
finite prefixes of such plays ending in an observation. Let $\gamma(\Plays(G))$
be the set of concrete paths of all plays in the game, and $\gamma(\Prefs(G))$ be
the set of all finite prefixes of all concrete paths.

An \emph{observation-based strategy for Eve} is a function from finite prefixes of
plays to actions, \ie $\lambda_\exists : \Prefs(G) \to \Sigma$. A play $\psi =
o_0 \sigma_0 o_1 \sigma_1 \dots$ is \emph{consistent} with $\lambda_\exists$ if
$\sigma_i = \lambda_\exists(\psi[..i])$ for all $i \ge 0$.  An
\emph{observation-based strategy for Adam} is a function $\lambda_\forall :
\Prefs(G) \times \Sigma \to \Obs$ such that for any prefix $\rho = o_0 \sigma_0
\dots o_n \in \Prefs(G)$ and action $\sigma$, $\lambda_\forall(\rho, \sigma) \cap
\post_\sigma(\rho[n]) \neq \emptyset$. A play $\psi = o_0 \sigma_0 o_1 \sigma_1
\dots$ is consistent with $\lambda_\forall$ if $o_{i+1} = \lambda_\forall(\psi[..i],
\sigma_i)$ for all $i\ge 0$. A \emph{concrete strategy for Adam} is a
function $\mu_\forall : \gamma(\Prefs(G)) \times \Sigma \to Q$ such that for any
concrete prefix $\pi = q_0 \sigma_0 \dots q_n \in \gamma(\Prefs(G))$ and action
$\sigma$, $\mu_\forall(\pi, \sigma) \in \post_\sigma(\{\pi[n]\})$. A play $\psi
= o_0 \sigma_0 o_1 \sigma_1 \dots$ is consistent with $\mu_\forall$ if there
exists a concrete path $\pi \in \gamma(\psi)$ such that $\mu_\forall(\pi[..i],
\sigma_i) = \pi[i+1]$ for all $i \ge 0$.

An observation-based strategy for Eve $\lambda_{\exists}$ can be encoded into a
\emph{finite Mealy machine} if there is a finite set $M$, an element $m_0 \in
M$, and functions $\alpha_u:M \times \Obs \to M$ and $\alpha_o:M \times \Obs \to
\Sigma$ such that for any play prefix $\psi = o_0 \sigma_0 \dots o_n$ we have
$\sigma_i = \lambda_{\exists}(\psi) = \alpha_o(m_n,o_n)$, where $m_n$ is defined
inductively by $m_{i+1} = \alpha_u(m_i,o_i)$ for $i\geq 0$. Similarly, an
observation-based strategy for Adam $\lambda_{\forall}$ can be encoded into a
finite Mealy machine if there is a finite set $M$, an element $m_0 \in M$,
and functions $\alpha_u:M \times \Obs \times \Sigma \to M$ and $\alpha_o:M
\times \Obs \times \Sigma \to \Obs$ such that for any play prefix ending in an
action $\psi=o_0 \sigma_0 \dots o_n \sigma_n$, we have $o_{i+1} =
\lambda_{\forall}(\psi) = \alpha_o(m_n,o_n,\sigma_n)$, where $m_n$ is defined
inductively by $m_{i+1} = \alpha_u(m_i,o_i,\sigma_i)$ for $i \ge 0$. In both
cases we say the observation-based strategy has \emph{memory} $|M|$.  An
observation-based strategy (for either player) with memory $1$ is
\emph{positional}.

\begin{remark}\label{rem:lim-obs-strats}
Note that for any concrete strategy $\mu$ for Adam there is a unique
observation-based strategy $\lambda_\mu$ for him such that all plays consistent
with $\mu$ are consistent with $\lambda_\mu$. Conversely there may be several,
but possibly no, concrete strategies that correspond to a single
observation-based strategy.  In games of limited observation there is guaranteed
to be at least one concrete strategy for every observation-based strategy.
\end{remark}

\paragraph*{Winning an MPG}
Given a threshold $\nu \in \mathbb{R}$, we say a play $\psi$ is \emph{winning
for Eve} if $\uMP(\pi) \geq \nu$ for all concrete paths $\pi \in
\gamma(\psi)$, otherwise it is \emph{winning for Adam}. Given $\nu$, one can
construct an equivalent game in which Eve wins if and only if
$\uMP(\pi) \geq 0$ if and only if she wins the original game, so
without loss of generality we will assume $\nu = 0$.  A strategy $\lambda$ is
\emph{winning} for a player if all plays consistent with $\lambda$ are winning
for that player. We say that a player \emph{wins} $G$ if (s)he has a winning
strategy.

\begin{remark}
	It was shown in~\cite{ddgrt10} that in MPGs with partial observation
	where finite-memory strategies suffice Eve wins the $\uMP$
	version of the game if and only if she wins the $\oMP$ version.
	As the majority of games considered in this paper only require finite
	memory, we can take either definition. For simplicity and consistency
	with Section~\ref{sec:liminf} we will use $\uMP$.
\end{remark}

\paragraph*{Non-zero-sum reachability games}
A \emph{reachability game} $G = ( Q, q_I, \Sigma, \Delta, \teve, \tadam
)$ is a tuple where $Q$ is a (not necessarily finite) set of states;
$\Sigma$ is a finite set of actions; $\Delta \subseteq Q \times \Sigma
\times Q$ is a finitary transition function (that is, for any $q \in Q$ and
$\sigma \in \Sigma$ there are finitely many $q' \in Q$ such that $(q,\sigma,q') \in
\Delta$); $q_I \in Q$ is the initial state; and $\teve,\tadam \subseteq Q$
are the terminating states. The game is played as follows. We place a token on
$q_I \in Q$ and start the game. Eve chooses an action $\sigma \in \Sigma$ and
Adam chooses a $\sigma$-successor of the current state as determined by
$\Delta$. The process is repeated until the game reaches a state in $\teve$ or
$\tadam$. In the first case we declare Eve as the winner whereas the latter
corresponds to Adam winning the game. Notice that the game, in general, might
not terminate, in which case neither player wins. Notions of plays and
strategies in the reachability game follow the definitions for mean-payoff
games, however we extend plays to include finite paths that end in $\teve \cup
\tadam$.

\section{Undecidability of Liminf Games}\label{sec:liminf}
Mean-payoff games with partial observation were extensively studied
in~\cite{ddgrt10}. In that paper the authors show that, with the mean payoff
condition defined using $\uMP$ and $>$, determining whether Eve has a winning
observation-based strategy is undecidable and when defined using $\oMP$ and
$\ge$, strategies with infinite memory may be necessary. The analogous, and more
general, questions using $\uMP$ and $\ge$ were left open. In this section we
answer these questions, showing that both results still hold.

\begin{proposition}
	There exist MPGs with partial observation for which Eve requires
	infinite-memory observation-based strategies to ensure $\uMP
	\geq 0$.
\end{proposition}
\begin{proof}
	Consider the game $G$ in Figure~\ref{fig:inf-mem}.  We will show that
	Eve has an infinite-memory observation-based strategy to win this game,
	but no finite-memory observation-based strategy.

	Consider the observation-based
	strategy that plays (regardless of the witnessed
	observations) $ab a^2 b a^3 b a^4 b \dots$ As $b$ is played infinitely
	often by this strategy, the only concrete paths consistent with it
	are $ \pi = q_0 q_1^\omega \text{ and } \pi = q_0 \cdot q_1^k
	\cdot q_2^l \cdot q_3^\omega$ for non-negative integers $k,l$.  In the
	first case we see that $\frac{1}{n} w(\pi[..n]) \to 0$ as $n \to
	\infty$, and for all paths matching the second case we have $\frac{1}{n}
	w(\pi[..n]) \to 1$ as $n \to \infty$. Thus $\uMP \geq 0$ and so the
	strategy is winning.

	Now suppose Eve has a finite-memory observation-based winning strategy
	for $G$.
	%Consider the observation-based strategy for Adam that ensures
	%the game remains in $\{q_1,q_2\}$.
	We will define a concrete strategy for Adam such that a concrete path
	with negative mean payoff and consistent with both strategies exists.
	The strategy for Adam is such that the game remains in $\{q_1,q_2\}$.
	The resulting play can now be seen as choosing a word $w \in
	\{a,b\}^\omega$, but as Eve's strategy has finite memory, this word must
	be ultimately periodic, that is $w = w_0 \cdot v^\omega$ for words $w_0,
	v \in \{a,b\}^*$. We now describe the concrete strategy for Adam.
	If $w$ contains finitely many $b$'s then Adam moves to $q_2$
	on the final $b$ and $\frac{1}{n} w(\pi[..n]) \to -1$ as $n \to \infty$.
	Otherwise Adam remains in $q_1$ and $\frac{1}{n} w(\pi[..n]) \to
	-\frac{m}{|v|}$ as $n\to \infty$ where $m$ is the number of $b$'s in
	$v$.
\end{proof}

\begin{theorem}
\label{thm:liminf}
	Let $G$ be an MPG with partial observation. Determining whether Eve has
	an observation-based strategy to ensure $\uMP \ge 0$ is
	undecidable.
\end{theorem}

The proof of this result is based on a similar construction to the one used in
the proof of Proposition~\ref{pro:superThm}, so we defer it to
Section~\ref{sec:proof-liminf}.

\section{Observable Determinacy}

\begin{figure}
\begin{minipage}[t]{.49\linewidth}
\begin{center}
\begin{tikzpicture}[node distance=1.2cm]
	\node[state,initial,fill=blue!20](A){$q_0$};
	\node[state,fill=yellow!30](B)[right=of A, yshift=1cm]{$q_1$};
	\node[state,fill=yellow!30](C)[right=of A, yshift=-1cm]{$q_2$};
	\node[state,fill=green!20](D)[right=of B, yshift=-1cm]{$q_3$};

	\node[fit=(A)]{};
	\node[fit=(D)]{};
	\node[fit=(B) (C)]{};

	\path
	(B) edge[bend right] node[el,swap] {$a$,-1} (A)
	(C) edge[bend left] node[el] {$b$,-1} (A)
	(A) edge[bend right] node[el,pos=0.6]{$\Sigma$,-1} (B)
	(A) edge[bend left] node[el,swap,pos=0.6] {$\Sigma$,-1} (C)
	(B) edge node[el] {$b$,-1} (D)
	(C) edge node[el, swap] {$a$,-1} (D)
	(D) edge[loop, out=30, in=-30,looseness=6] node[el] {$\Sigma$,+1} (D);
\end{tikzpicture}
\caption{A non-determined MPG with limited observation ($\Sigma = \{a,b\}$)}
\label{fig:purendet}
\end{center}
\end{minipage}
\hfill
\begin{minipage}[t]{.49\linewidth}
\begin{center}
\begin{tikzpicture}[node distance=1.2cm]
	\node[state,initial,fill=blue!20](A){ $q_0$};
	\node[state,fill=yellow!30](B)[right=of A, yshift=1cm]{$q_1$};
	\node[state,fill=yellow!30](C)[right=of A, yshift=-1cm]{$q_2$};
	\node[state,fill=green!20](D)[right=of B, yshift=-1cm]{$q_3$};

	\node[fit=(A)]{};
	\node[fit=(D)]{};
	\node[fit=(B) (C)]{};

	\path
	(B) edge[loop, out=30, in=-30,looseness=6] node[el,pos=0.45] {$a$,0}
	node[el,pos=0.50]{$b$,-1} (B)
	(C) edge[loop, out=30, in=-30,looseness=6] node[el] {$a$,-1} (C)
	(A) edge node[el,pos=0.6]{$\Sigma$,0} (B)
	(A) edge node[el,swap,pos=0.6] {$\Sigma$,0} (C)
	(B) edge node[el,swap,inner sep=0.7] {$b$,-1} (C)
	(C) edge[bend left] node[el,pos=0.75] {$b$,0} (D)
	(D) edge[loop, out=30, in=-30,looseness=6] node[el] {$\Sigma$,+1} (D);
\end{tikzpicture}
\caption{A limited-observation MPG in which Eve requires infinite memory
	to win}
\label{fig:inf-mem}
\end{center}
\end{minipage}
\end{figure}

One of the key features of MPGs with full observation is that they are
determined, that is, it is always the case that one player has a winning
strategy. This is not true in games of partial or limited observation as can be
seen in Figure~\ref{fig:purendet}. Any concrete strategy of Adam reveals to Eve
the successor of $q_0$ and she can use this information to play to $q_3$.
Conversely Adam can defeat any strategy of Eve by playing to whichever of $q_1$
or $q_2$ means the play returns to $q_0$ on Eve's next choice (recall Eve cannot
distinguish $q_1$ and $q_2$ and must therefore choose an action to apply to the
observation $\{q_1,q_2\}$). This strategy of Adam can be encoded as an
observation-based strategy: ``from $\{q_1,q_2\}$ with action $a$ or $b$, play to
$\{q_0\}$''. It transpires that any
such counter-play by Adam is always encodable as an observable strategy. We
formalize these claims in the sequel.

Let us recall the definition of the Borel hierarchy of sets. For a detailed
description of both the hierarchy and its properties we refer the reader
to~\cite{kechris95}.

\begin{definition}[Borel hierarchy \& (co-)Suslin sets]
	For a (possibly infinite) alphabet $A$, let $A^\omega$ and $A^*$ denote
	the set of infinite and finite words on $A$, respectively. The
	\emph{Borel hierarchy} is inductively defined as follows.
	\begin{itemize}[nolistsep]
		\item $\Sigma_1^0 = \{ W \cdot A^\omega \st W \subseteq A^*\}$
			is the set of open sets.
		\item For all $n \ge 1$, $\Pi_n^0 = \{A^\omega \setminus L \st L
			\in \Sigma_n^0\}$ consists of the complement of sets in
			$\Sigma_n^0$.
		\item For all $n \ge 1$, $\Sigma_{n+1}^0 = \{ \bigcup_{i \in
				\mathbb{N}} L_i \st \forall i \in \mathbb{N} :
			L_i \in \Pi_n^0\}$ is the set obtained by countable
			unions of sets in $\Pi_n^0$.
		\item Finally, we write $\Delta_n^0 = \Sigma_n^0 \cap \Pi_n^0$,
			for all $n \ge 0$.
	\end{itemize}
	The first level of the \emph{Projective hierarchy} consists of
	$\Sigma_1^1$ (Suslin) sets, which are those whose preimage is a Borel
	set, \ie all sets that can be defined as a projection of a Borel set,
	and $\Pi_1^ 1$ (co-Suslin) sets: those sets whose complement is the
	image of a Borel set.
\end{definition}

\subsection{(Full-observation) Determinacy}
Let us first consider MPGs with full observation and recall the well-known
determinacy result that applies to them. Note that in games with full
observation a play in fact corresponds to a unique infinite concrete path.
Furthermore, the distinction between observation-based and concrete strategies
is unnecessary. For clarity, in the remaining of this section we speak of
\emph{concrete plays} in full-observation games and \emph{abstract plays} in
partial-observation games. Given a strategy $\lambda_\exists$ for Eve and a
strategy $\lambda_\forall$ for Adam in an MPG, we denote by
$\out(\lambda_\exists,\lambda_\forall)$ the unique play consistent with both
strategies.

\begin{proposition}
	In every MPG with full observation exactly one of the following
	assertions holds.
	\begin{enumerate}[nolistsep]
		\item There exists a strategy $\lambda_\exists$ for Eve such
			that, for all strategies $\lambda_\forall$ for Adam, the
			concrete
			play $\out(\lambda_\exists,\lambda_\forall)$ is winning
			for Eve.
		\item There exists a strategy $\lambda_\forall$ for Adam such
			that, for all strategies $\lambda_\exists$ for Eve, the
			concrete
			play $\out(\lambda_\exists,\lambda_\forall)$ is winning
			for Adam.
	\end{enumerate}
\end{proposition}
The proof of the above determinacy result follows from the fact that the set of
winning plays in any MPG is a Borel set. More precisely,
the statement that the limit inferior of a given sequence $(a_n)_{n \in
\mathbb{N}}$ is non-negative is a $\Pi^0_3$-statement (for every $k$
there exists a $t$ such that for all $n \geq t$ \ $a_n \geq - 2^{-k}$).
Similarly, for the limit superior we get a $\Sigma^0_2$-statement. Hence, by
Borel determinacy~\cite{martin75}, all mean-payoff games with full observation
are determined.

\subsection{Determinacy \& partial-observation games}
In MPGs with partial observation, authors usually focus on observation-based
strategies for Eve and concrete strategies for Adam. Using this asymmetric point
of view, we will now state the well-known non-determinacy of games with partial
observation. Given an observation-based strategy $\lambda_\exists$ for Eve and a
concrete strategy $\mu_\forall$ for Adam in an MPG with partial observation, we
denote by $\out(\lambda_\exists, \mu_\forall)$ the unique abstract play consistent
with both strategies. Remark that here we can no longer assume a play is one
unique concrete path. Furthermore, recall that an abstract play is winning (for
a player) if all its concretizations are winning (for the player)
\begin{proposition}
	There are MPGs with limited observation for which none of the following
	assertions hold.
	\begin{enumerate}[nolistsep]
		\item There exists an observation-based strategy
			$\lambda_\exists$ for Eve such that, for all concrete
			strategies $\mu_\forall$ for Adam, the abstract play
			$\out(\lambda_\exists,\mu_\forall)$ is winning for Eve.
		\item There exists a concrete strategy $\mu_\forall$ for Adam
			such that, for all observation-based strategies
			$\lambda_\exists$ for Eve, the abstract play
			$\out(\lambda_\exists,\mu_\forall)$ is winning for Adam.
	\end{enumerate}
\end{proposition}
One such game is shown in Figure~\ref{fig:purendet}. Intuitively,
the second statement is too strong for it to be implied by the negation of
the first. Indeed, the game remains asymmetric---in favour of Adam---just
as long as Eve does not know the concrete path corresponding to the current play
prefix.  This is not the case in the second statement because of the order of
quantifications over the strategies and the fact Adam is using a concrete
strategy. That is, since she knows his concrete strategy and the current play
prefix, she knows the current concrete state of the game as well.

We presently show that if one considers the ``more symmetrical'' statements in
which both players use observation-based strategies, then we recover
determinacy.
\begin{theorem}[Observable determinacy]
\label{thm:det}
	In every MPG with limited observation (defined with $\uMP$ or $\oMP$)
	exactly one of the following assertions holds.
	\begin{enumerate}[nolistsep]
		\item There exists an observation-based strategy
			$\lambda_\exists$ for Eve such that, for all
			observation-based strategies $\lambda_\forall$ for Adam,
			the abstract play
			$\out(\lambda_\exists,\lambda_\forall)$ is winning for
			Eve.
		\item There exists an observation-based strategy
			$\lambda_\forall$ for Adam such that, for all
			observation-based strategies $\lambda_\exists$ for Eve,
			the abstract play
			$\out(\lambda_\exists,\lambda_\forall)$ is winning for
			Adam.
	\end{enumerate}
\end{theorem}
In what follows we will first show how to construct a non-deterministic
mean-payoff automaton that recognizes as its language the set of all concrete
plays that are
winning for Adam in a given MPG with limited observation. We then show
that the language of the automaton is a Borel set.  The result will thus follow
from Lemma \ref{lemma:gamestoautomata}, Corollaries \ref{corr:automatasigma4}
and~\ref{cor:limsup-sigma4}, the fact that Borel sets are closed under
complement, and Borel determinacy~\cite{martin75}.

Determinacy usually enables to simplify proofs.
Without any sort of determinacy, game reductions
can become tedious and confusing. However, with determinacy, we can simply
transfer winning strategies for both players between games and that directly
implies both players win in one game if and only if they win in the game we reduce
to (from). We remark that although
our results in Section~\ref{sec:stratTransfer} already imply
that the games we consider from then onwards are observably determined, the
above result is more general. That is to say, all partial-observation MPGs which
do not fit into the classes we consider later in this work, are observably
determined.
It is also worth noting that observation-based strategies for Adam only
really make sense for games with limited observation, since in a general
partial-observation game he may be able to choose an observation which yields an
abstract path with an empty set of concretizations. (The latter is not possible
in a limited-observation game.) Since, for a given partial-observation game, the
equivalent limited-observation game may be of size exponential w.r.t. to the
original game, these kind of determinacy results may be useful in instances
where one is interested in decidability of game-related problems but not
necessarily when interested in establishing complexity bounds.

\subsection{Borelness of losing plays}
In~\cite{cd10} the authors consider parity objectives and
show that, given a game with partial observation, one
can construct a non-deterministic automaton (with the negation of the game's
objective as acceptance condition) that recognizes the set of plays that are
winning for Adam. We adapt their construction for MPGs. Let $G = ( Q,
q_I,\Sigma,\Delta,w,\Obs )$ be a limited-observation liminf
(respectively, limsup) MPG. We construct a \emph{mean-payoff automaton}
$\mathcal{A} = ( Q, q_I,A,T,c )$ where:
\begin{itemize}[nolistsep]
	\item $A = \Sigma \times \Obs$ is the alphabet of the automaton,
	\item $T = \{ (p,(\sigma,o),q) \st (p,\sigma,q) \in \Delta \land q \in
		o\}$ is the transition relation, and
	\item $c$ is a weight function such that, $(p,(\sigma,o),q) \mapsto
		-w(p,\sigma,q)$.
\end{itemize}
A \emph{run} of $\mathcal{A}$ over an infinite word $\alpha = a_0 a_1 \dots \in
A^\omega$ is an infinite sequence $\rho = q_0 a_0 q_1 a_1 \dots$ such that $q_0
= q_I$ and $(q_i,a_i,q_{i+1}) \in T$ for all $i \ge 0$. We say $\rho$ is
\emph{accepting} if the limit superior (resp. limit inferior) of the sequence
$(c(q_i,a_i,q_{i+1}))_{i \in \mathbb{N}}$ is strictly positive.  Depending on
its acceptance condition, we say the constructed machine is a limsup
(resp. liminf) mean-payoff automaton. Finally, the \emph{language} of a
mean-payoff automaton \emph{recognizes} is the set $\{ \sigma_0 o_0 \sigma_1 o_1
\dots \st \text{ there is an accepting run of } \mathcal{A} \text{ over }
(\sigma_0,o_0) (\sigma_1,o_1) \dots \}$.

Clearly, if we write $\mathcal{L}_\mathcal{A}$ for
the language of an automaton $\mathcal{A}$ constructed for an MPG
with limited observation $G = ( Q,q_I,\Sigma,\Delta,w,\Obs )$,
then the set $\{q_I\} \cdot \mathcal{L}_A \subseteq \Plays(G)$
is the set of all plays in $G$ which are \textbf{not winning for
Eve}. Intuitively, $\mathcal{A}$ receives the choice of action $\sigma$ for Eve
and observation $o$ for Adam and then ``guesses'' the actual state chosen by a
concrete strategy for Adam, thus constructing all concretizations of a play in
parallel and accepting if one of them is losing for Eve (\ie, winning for
Adam).
\begin{lemma}
\label{lemma:gamestoautomata}
The set of winning plays in a limited-observation MPG is recognizable by a
non-deterministic mean-payoff automaton.
\end{lemma}

We will now show that the language of any non-deterministic mean-payoff
automaton is a Borel set.

\paragraph*{Liminf mean-payoff automata}
Recall that the statement that the limit inferior of a given sequence is a
$\Pi^0_3$-statement. Thus, any set recognized by a deterministic liminf
mean-payoff automaton is $\Sigma^0_4$. The jump from $\Pi^0_3$ is due to the
strictness of the inequality (if there exists $b > 0$ such that
$\uMP \ge b$, then $\uMP > 0$). Moving to non-deterministic liminf
mean-payoff automata, however, adds an existential quantification over all
runs---and hence looks like it could go as high as $\Sigma^1_1$. (Whether or not all
games with $\Sigma^1_1$ winning-play sets are determined is independent
of \textrm{ZFC}. A positive answer follows, \eg, from the existence of a
measurable cardinal~\cite{ms88}.) However, in the following we shall see that
non-deterministic liminf mean-payoff automata still only recognize
$\Sigma^0_4$-sets.

\begin{proposition}
\label{theo:automatapi3}
The following are equivalent for a non-deterministic liminf mean-payoff automaton:
\begin{enumerate}[nolistsep]
\item There exists a run over $\alpha$ with non-negative liminf mean-payoff.
\item For any $k$ there exists a run $p_k$ over $\alpha$ and a position $t_k \in
	\mathbb{N}$ such
	that the mean payoff along $p_k$ never falls below $-2^{-k}$ after
	position
	$t_k \in \mathbb{N}$.
\end{enumerate}
\end{proposition}

Essentially, the difference between (1.) and (2.) is that the existential
quantifier over the runs is moved inwards. In particular, it is obvious that
(1.) implies (2.), but the converse direction is non-trivial. The basic idea of
the proof is that we construct a new run $p$ from the runs $p_k$ by always
following some run for some time, and then switching to a run for higher $k$,
and so on. We are faced with two problems: We can only switch from a run to
another if they are at the same state of the automaton at the same time, so we
might get \emph{stuck} in a run which never meets another run for higher $k$.
Moreover, a run $p_k$ could at some position $t$ have much higher current mean
payoff than a run $p_{k'}$ with $k > k'$, and proceed to lose a lot of
payoff---which $p_{k'}$ could not afford.

Thus, in order to construct our run $p$, we need to make sure that we always
have the option available to switch to a suitable run for higher $k$ at some
position where the two runs have very similar current mean payoff. The existence of
suitable collections will be proven by iterative applications of Ramsey's
theorem:

\begin{theorem}[Infinite Ramsey's theorem]
Let $\mathcal{P}(\mathbb{N})_r$ denote the set of $r$-element subsets of
$\mathbb{N}$. Then for any \emph{colouring function}
$c : \mathcal{P}(\mathbb{N})_r \to \{0,1,\dots,\kappa\}$, where $\kappa \in \mathbb{N}$,
there exists an infinite subset $H \subseteq \mathbb{N}$ such that
for any two $A, B \in \mathcal{P}(\mathbb{N})_r$ with $A \subseteq H$, $B
\subseteq H$ we find that $c(A) = c(B)$. Such an $H$ is called $c$-homogenous.
\end{theorem}

\begin{proof}[Proof of Proposition~\ref{theo:automatapi3}]
Assume that for any $k$ there exists a run $p_k$ and a position $t_k$ such that the
mean payoff along $p_k$ never falls below $-2^{-k}$ after position $t_k$. Let
states in the automaton be labelled $0$ to $n$. W.l.o.g., assume that the payoff
values are from $[-1,1]$. For $x \in [-1,1]$ and $\ell \in \mathbb{N}$, let
$b_{x,\ell} := \lceil 2^\ell(x+1)\rceil$. Note that $b_{x,\ell} \in
\{0,\ldots,2^{\ell+1}\}$. Let $\langle \ , \ \rangle : \mathbb{N} \times \mathbb{N}
\to \mathbb{N}$ denote a standard pairing function, \ie an encoding of a pair of
natural numbers into one single natural number.

To use Ramsey's theorem, one must colour subsets of $\mathbb{N}$, not tuples.
Since we will use it on sets $\{k,t\}$ with $k$ the index of a run and $t$ a
position, we must somehow decide which number in a given set of size $2$ is the run
index and which one is the position. Since for any position $t \ge t_k$
the mean payoff of $p_k$ never falls below $-2^{-k}$ after position $t$, the larger
number in a pair will always be understood as the position while the smaller will be
the run index.

We iteratively define colourings $c_i$ of $2$-element subsets of $\mathbb{N}$,
to which we apply Ramsey's theorem in order to obtain $c_i$-homogenous sets
$H_i$ and derived infinite sets $S_i$. For all $k < t$, let $c_0(\{k,t\}) :=
\langle v, b_{x, 1}\rangle$ where $v$ is the state the run $p_k$ is in at position
$t$, and $x$ is the current mean payoff of $p_k$ at position $t$. Let $H_0$ be an
infinite $c_0$-homogenous set. Let $S_0 := H_0$. Once we have obtained $S_i$,
let $(m^i_{j})_{j \in \mathbb{N}}$ be a monotone (increasing) sequence
enumerating $S_i$.  Then let $c_{i+1}(\{k,t\})$, again for $k < t$, be $\langle
v, b_{x,i+2}\rangle$ where $v$ is the state the run $p_{m^i_k}$ is in at position
$t$ and $x$ is the current mean payoff of $p_{m_k^i}$ at position $t$. Let $H_{i+1}$
be an infinite $c_{i+1}$-homogenous set. Let $S_{i+1} = \{m^i_k \st k \in
H_{i+1}\}$.

This construction ensures that for all $i \in \mathbb{N}$, for all $k_1, k_2 \in
S_i$, and for all sufficiently large $t \in H_i$, the runs $p_{k_1}$ and
$p_{k_2}$ will be at the same state at position $t$ and their current mean payoff at
position $t$ will differ by at most $2^{-i-1}$. If $i > 0$, by sufficiently large we
mean any $t$ larger than the indices $j,\ell$ we assign to $k_1$ and $k_2$ in
the construction of $H_i$, \ie $t > m^i_j,m^i_\ell$ where $k_1 = m^i_j$ and $k_2
= m^i_\ell$.  Since the sequence $(m^i_j)_{j \in \mathbb{N}}$ is monotone, it
suffices to take $t > \max\{k_1,k_2\}$. (The latter also allows us to claim the
property holds for $i = 0$.) Finally, also note that $S_{i+1} \subseteq S_i$.

The run $p$ we need to construct will first follow some $p_{k_0}$ with $k_0 \in
S_0$ until sufficient time $s_0 \in H_0$ has passed, then switch to some
$p_{k_1}$ with $k_1 > k_0$ and $k_1 \in S_1$, again until sufficient (total)
time $s_1 \in H_1$ has passed, then switch to $p_{k_2}$ with $k_2 > k_1$ and
$k_2 \in S_2$, and so forth.

It remains to specify what \emph{sufficient time} means for the position $s_i$, and to
show that this condition ensures that the mean payoff of $p$ is non-negative. For
the latter, we will ensure that after position $s_i$ the current mean payoff of $p$
never again drops below $-2^{-i+1}$. The sufficient condition for $s_i$ will
include the sufficiency condition for any runs with indices from $S_i$ having
the same current vertex and current mean-payoff difference at most $2^{-i-1}$.
Moreover, we need that $s_i \geq t_{k_{i+1}}$.

Let us now consider the current mean payoff of $p$ at some position $t$ with $s_i
\leq t < s_{i+1}$. By summing up the loss of mean payoff through the changes, we
see that the current mean payoff of $p$ differs by at most $2^{-2}\frac{s_0}{t}
+ \ldots +  2^{-i -2}\frac{s_i}{t}$ from that of $p_{k_i}$, which in turn is at
least $-2^{-k_i}$. Note that, since we have chosen our indices so that $k_i \le
k_{i+1}$ for all $i \in \mathbb{N}$, we have that $k_i \ge i$ and therefore
$-2^{-k_i} \geq -2^{-i}$.  Thus, once $s_0, \ldots, s_{i-1}$ have been chosen,
we just need make sure that $s_i$ is large enough so that:
\begin{itemize}[nolistsep]
	\item $s_i \ge t_{k_{i+1}}$,
	\item $s_i \ge \max\{k_i,k_{i+1}\}$, and finally
	\item $\sum_{j=0}^{i} 2^{-j-1}\frac{s_j}{s_i} \leq 2^{-i-1}$ so that the
		mean payoff of $p$ never again drops below $-2^{-i+1}$.
\end{itemize}
The first two items are trivial. For the third one, note that, as the left hand
side goes to $0$ for $s_i \to \infty$, this can always be ensured by staying
increasingly longer with each run we switch to.
\end{proof}

\begin{corollary}
\label{corr:automatasigma4}
Any set recognized by a liminf mean-payoff automaton is $\Sigma^0_4$.
\begin{proof}
	Using Proposition \ref{theo:automatapi3}, it suffices to argue that the
	second equivalent condition is $\Pi^0_3$. This in turn follows from the
	observation that for fixed $k, t \in \mathbb{N}$ the condition
	\emph{there exists a run whose mean payoff never falls below $-2^k$
	after position $t$} is by Weak K\"onig's Lemma a $\Pi^0_1$ condition.
\end{proof}
\end{corollary}

\paragraph*{Limsup mean-payoff automata}
We now show an analogue of Corollary~\ref{corr:automatasigma4} holds for limsup
mean-payoff automata. Once more, to simplify the argument, we focus on the
non-strict acceptance condition. We show languages recognized by such automata are
$\Pi^0_3$ and, thus, those recognized by limsup mean-payoff automata are
$\Sigma_0^4$.

\begin{proposition}
\label{theo:automatalimsup}
The following are equivalent for a non-deterministic limsup mean-payoff automaton:
\begin{enumerate}[nolistsep]
	\item There exists a run over $\alpha$ with non-negative limsup mean payoff.
	\item For any $k$ there exists a run $p_k$ over $\alpha$ such that for
		all positions $t$ there exists a position $t' > t$ such that the mean
		payoff of $p_k$ at position $t'$ is at least $-2^{-k}$.
\end{enumerate}
\begin{proof}
That (1.) implies (2.) follows immediately from the definition of limsup mean
payoff and the fact that the
witnessing run from (1.) is also a witness for all $k$ in (2.). That
(2.) implies (1.) will be shown using Ramsey's theorem, similar to the argument
in the proof of Proposition~\ref{theo:automatapi3}. We do not need iterative
applications here, though.

We define a colouring $c$ by setting $c(\{t,k\}) = v$ where $k < t$ and the run
$p_k$ is in the state $v$ at position $t$, and obtain an infinite $c$-homogenous set
$S$. If $k, k' \in S$, then for any position $t \in S$ with $t > \max \{k,k'\}$ the
runs $p_k$ and $p_{k'}$ are at the same state at position $t$, and hence we are
allowed to switch from one run to the other. Let $(k_m)_{m \in \mathbb{N}}$ be a
monotone sequence enumerating $S$. We first follow $k_0$ for a while, then
switch to $k_1$, and so on. This constructs the witnessing run $p$.

By assumption, the run $p_{k_0}$ will eventually reach a current mean payoff of
at least $- 2^{-k_0}$ at some position $t_0$. We pick some $t'_0 \in S$ with $t'_0 >
\max \{t_0, k_0, k_1\}$, and follow $p_{k_0}$ until position $t'_0$, and then switch
to the run $p_{k_1}$. At the time of the switch, there is some $c_0$ such that
the current mean payoff of $p_{k_0}$ is not more than $c_0$ below the current
mean payoff of $p_{k_1}$. (Our intention is to make the difference $c_0$, along
with all future ones, disappear by staying longer and longer with runs we switch
to.)
This implies that at positions $t > t'_0$ (but prior to
the next switch) the mean payoff of $p$ is at least the mean payoff of $p_{k_1}$
minus $\frac{t'_0}{t}c$, for some $c$. 
We pick $t_1 > t'_0$ such that
$\frac{t'_0}{t_1}c \leq 2^{-k_1}$,
and then some $t_1' \geq t_1$ such that
$p_{k_1}$ at position $t_1'$ has a mean payoff of at least $-2^{-k_1}$. Then $p$ has
a current mean payoff of at least $-2^{-k_1 + 1}$ at position $t_1'$. Then let
$t_1'' \in S$ such that $t_1'' > \max \{t_1', k_1, k_2\}$. The run $p$ follows
$p_{k_1}$ until position $t_1''$, and then switches to $p_{k_2}$. Again, we will
follow $p_{k_2}$ long enough so that the accumulated difference of the mean
payoff caused by the switches is small enough, and then until $p_{k_2}$ realizes
its bound next, and then switch to $p_{k_3}$ at the next possible chance, and so
on. As we keep reaching mean payoff values closer and closer to $0$, the limsup
mean payoff of $p$ is non-negative, as intended.
\end{proof}
\end{proposition}

To obtain an analogue of Corollary~\ref{cor:limsup-sigma4} here, we must still
argue that the statement $\phi$: \emph{there exists a run over $\alpha$ whose mean
payoff is at least $-2^k$ infinitely often}, is Borel. In the sequel we show how
to encode all run prefixes of the mean-payoff automaton over a given word
$\alpha$ into a DAG. The DAG is constructed to have special edges witnessing the
existence of a run prefix with mean payoff of at least $-2^k$. We show the DAG
has an infinite path that traverses such edges infinitely often if and only if
$\phi$ holds. To conclude, we then argue the set of all such DAGs is
$\omega$-regular and thus Borel.

\begin{definition}
Given an automaton $M$ with $n$ states, an input $\alpha$ and some precision
parameter $k \in \mathbb{N}$, we define a directed acyclic graph (DAG, for short)
$D_k(\alpha)$ over
$\{1,\ldots,n\} \times \mathbb{N}$ with the structural constraint that there are only
edges from $(i,\ell)$ to $(j,\ell+1)$ or from $(i,\ell)$ to $(j,\ell+2)$. We fix some
sequence $(t_\ell)_{\ell \in \mathbb{N}}$ such that $\frac{t_\ell}{t_{\ell+1}} \leq
2^{-k-1}$. There is an edge from $(i,\ell)$ to $(j,\ell+1)$ if there is an run
of $M$ starting from state $i$, reading in $\alpha$ from position $t_\ell$ to
position $t_{\ell+1}$ and ending in $j$. There is an edge from $(i,\ell)$ to
$(j,\ell+2)$
if there is an run of $M$ starting from state $i$, reading in $\alpha$ from
position $t_\ell$ to position $t_{\ell+2}$ and ending in $j$, such that at some point
$t$ with $t_\ell < t < t_{\ell+2}$ the current mean payoff $x$ satisfies that
$\frac{t - t_\ell}{t}x - \frac{t_\ell}{t} \geq -2^{-k}$. We call the latter
\emph{long edges}.
\end{definition}

The following remark will be useful in the sequel to establish that languages
recognized by limsup mean-payoff automata are Borel. Intuitively, we will argue
that the set of DAGs containing infinite paths with infinitely many long edges
form a Borel set. 

\begin{remark}
\label{obs:dag}
	We can code the DAGs $D_k(\alpha)$ for a fixed $M$ into an infinite
	sequence over a finite alphabet, by letting the $l$-th symbol code which
	of the finitely many potential edges from some $(i,\ell)$ to
	$(j,\ell+1)$ and from $(i,\ell)$ to $(i,\ell+2)$ are available. 
\end{remark}

We will now sketch how to recognize such DAGs using a \emph{B\"uchi automaton}.
That is, an infinite-word non-deterministic automaton $(S,s_I,A,T,B)$ with
finite set of states $S$, initial state $s_I$, alphabet $A$, transition relation
$T \subseteq S \times A \times S$ and \emph{accepting transition set} $B
\subseteq T$. The notion of run is as for mean-payoff automata. We say a run of
a B\"uchi automaton is accepting if it contains infinitely many accepting
transitions, and define its language as for mean-payoff automata.

\begin{remark}\label{rem:dag-auto}
	There exists a finite non-deterministic B\"uchi automaton that reads in
	DAGs (represented as described in Remark~\ref{obs:dag}), and accepts
	exactly those DAGs that admit a path containing infinitely many long
	edges starting from $(i_0,0)$, where $i_0$ is the initial state of $M$.
	For instance, the B\"uchi automaton can be defined using the same
	state-space as the original mean-payoff automaton, with a transition
	from $i$ to $j$ on input edge $\left( (i,\ell) , (i,\ell') \right)$ if
	$j$ is reachable from $i$ in the original automaton. The transition is
	then marked as accepting if the edge being read is long.
\end{remark}

\begin{lemma}
\label{lem:limsup1}
If $D_k(\alpha)$ admits a path containing infinitely many long edges
starting from $(i_0,0)$, where $i_0$ is the initial state of
$M$, then $M$ has a run over $\alpha$ which has a mean payoff of at least
$-2^{-k}$ infinitely many times.
\begin{proof}
Any edge in $D_k(\alpha)$ is witnessed by a partial run of $M$, in such a
way that an infinite path through $D_k(\alpha)$ starting from $(i_0,0)$ gives rise
to a full run of $M$ on input $\alpha$. Given the condition for adding a long
edge, \ie of the form $(i,\ell) \Rightarrow (j,\ell+2)$, we note that at the
witnessing position $t$, the current mean payoff of $M$ is of the form
$\frac{t_\ell}{t}y + \frac{t - t_\ell}{t}x$, where $x$ is the mean payoff from position
$t$ to position $t_\ell$, and $y$ the mean payoff from the start to position $t$. As we
assume all payoffs to be from $[-1,1]$, this is bounded from below by $\frac{t -
t_\ell}{t}x - \frac{t_\ell}{t} \geq -2^{-k}$, hence the claim follows.
\end{proof}
\end{lemma}

\begin{lemma}
\label{lem:limsup2}
If $M$ has a run over $\alpha$ reaching a mean payoff of at least
$-2^{-k-1}$ infinitely often, then $D_k(\alpha)$ admits a path containing
infinitely many long edges starting from
$(i_0,0)$.
\begin{proof}
Following a run $p$ through $M$, we can construct an infinite path through
$D_k(\alpha)$ as follows: Start from $(i_0,0)$. If we are currently at $(i,\ell)$,
and there is an edge available to $(j,\ell+2)$ where $p$ is in state $j$ at position
$t_{\ell+2}$, take that edge. Else, take the edge to $(j',\ell+1)$, where $j'$ is the
state $p$ is in a position $t_{\ell+1}$ (by construction of $D_k(\alpha)$, the latter
always exists. We claim that if the mean payoff of $p$ is at least
$-2^{-k-1}$ infinitely often, then the former case occurs infinitely many times.

Assume that $p$ has mean payoff at least $-2^{-k-1}$ at position $t$ with
$t_{\ell+1} \leq t < t_{\ell+2}$. Let $p$ be in state $i$ at position $t_\ell$ and
at state $j$ at position $t_{\ell+2}$. We claim that $D_k(\alpha)$ has an edge from
$(i,\ell)$ to $(j,\ell+2)$. To see that, note that the mean payoff of $p$ at
position $t$ is of the form $x\frac{t-t_\ell}{t} + y\frac{t_\ell}{t}$, where $x$ is
the mean payoff from position $t_\ell$ to position $t$, and $y$ the mean payoff from the
start to position $t_\ell$. As $y \geq -1$, we can conclude that $x\frac{t-t_1}{t}
\geq -2^{-k-1} - \frac{t_\ell}{t}$. As $t \geq t_{\ell+1}$ and the constraint on
the choice of $(t_\ell)_{\ell \in \mathbb{N}}$ that $\frac{t_\ell}{t_{\ell+1}}
\leq 2^{-k-1}$, we in turn find that $x\frac{t-t_1}{t} \geq -2^{-k}$ -- hence
the claimed edge exists.

The only reason why we might be unable to choose such an edge from $(i,\ell)$ to
$(j,\ell+2)$ is if we choose some edge from $(i',\ell-1)$ to $(j',\ell+1)$
earlier. But that means that the availability of infinitely many such edges
implies that our construction will choose them, hence showing the claim.
\end{proof}
\end{lemma}

\begin{corollary}\label{cor:limsup-sigma4}
	Any set recognized by a limsup mean-payoff automaton is
	$\Sigma^0_4$.
\begin{proof}
	By Lemmas \ref{lem:limsup1} and~\ref{lem:limsup2} we find that for some
	input $\alpha$ and limsup mean payoff automaton $M$ the following are
	equivalent:
	\begin{enumerate}[nolistsep]
	\item For all $k$ there exists a run $p_k$ over $\alpha$ such that for all
		positions $t$ there exists a position $t' > t$ such that the mean payoff of
		$p_k$ at position $t'$ is at least $-2^{-k}$.
	\item For all $k$ the DAG $D_k(\alpha)$ has a path containing infinitely many
		long edges starting from $(i_0,0)$.
	\end{enumerate}
	By Proposition \ref{theo:automatalimsup}, the former is equivalent to
	$M$ accepting $\alpha$, and by Remarks~\ref{obs:dag}
	and~\ref{rem:dag-auto}, and~\cite{perrin04},
	the latter is a universal quantification over a $\Delta^0_3$-set, hence
	a $\Pi^0_3$-set. Thus, with the strict acceptance condition, a
	$\Sigma_0^4$ set.
\end{proof}
\end{corollary}

\section{Strategy Transfer}\label{sec:stratTransfer}
In this section we will construct a reachability game from an MPG with
limited observation in which winning strategies for either player are sufficient
(but not necessary) for observation-based winning strategies in the original MPG.

Let us fix a limited-observation MPG $G = ( Q, q_I, \Sigma, \Delta, w,
\Obs )$. We will define a reachability game on the weighted unfolding of
$G$.

\paragraph*{Belief functions}
Let $\calB$ be the set of (\emph{belief}) functions $f : Q \to \mathbb{Z} \cup
\{ +\infty, \bot \}$. Our intention is to use functions in $\calB$ to keep track
of the minimum payoff of all concrete paths ending in the given state. A
function value of $\bot$ indicates that the given state is not in the current
observation, and a function value of $+\infty$ is used to indicate to Eve that
the token is not located at such a state. Intuitively, $+\infty$ will allow our
reachability winning condition to include games where Adam wins by ignoring
paths going through the given state.

The \emph{support} of $f \in \calB$ is
$\supp(f) := \{q \in Q \st f(q) \neq \bot \}$.
We define a family of partial orders $\preceq_k \subseteq \calB \times \calB$
for all $k\in \mathbb{N}$. Formally, $f \preceq_k f'$ if:
\begin{itemize}[nolistsep]
	\item $\supp(f) = \supp(f')$ and
	\item $f(q)+k \leq f'(q)$ for all $q \in \supp(f)$
\end{itemize}
where $+ \infty + k = +\infty$.

\paragraph*{(Proper) successor functions}
Given two functions $f,f' \in \calB$, we say $f'$ a $\sigma$-successor of $f$ if:
\begin{itemize}[nolistsep]
	\item $\supp(f') \in \Obs$;
	\item $\supp(f') \subseteq \post_\sigma(\supp(f))$; and
	\item for all $q \in \supp(f')$ either
		\begin{itemize}[nolistsep]
			\item $f'(q) = \min\{f(q') +w(q',\sigma,q) \st q' \in
				\supp(f) \text{ and } (q',\sigma,q) \in \Delta\}$, or
			\item $f'(q) = +\infty$.
		\end{itemize}
\end{itemize}
Moreover, if $f'$ is a $\preceq_0$-minimal $\sigma$-successor of $f$, we say it
is a \emph{proper} $\sigma$-successor of $f$.

\paragraph*{Function-action sequences}
Let us denote by $\calF(G)$ the set of all sequences $f_0 \sigma_0 f_1 \dots
f_n \in (\calB \cdot \Sigma)^* \calB$ such that for all $0 \le i <
n$, $f_{i+1}$ is a $\sigma_i$-successor of $f_i$.  Observe that for each
$\phi = f_0 \sigma_0 \dots f_n \in \calF(G)$ there is
a unique abstract path $\supp(\phi) := o_0 \sigma_0 \dots o_n$ such that $o_i =
\supp(f_i)$ for all $i$. Conversely for each finite abstract path $\psi = o_0
\sigma_0 \dots o_n$ there may be many corresponding function-action sequences
$\supp^{-1}(\psi) := \{\phi \in \calF(G) \st \supp(\phi) = \psi \}$.  Of
particular interest are function-action sequences that are minimal with respect
to $\preceq_0$. Given a finite abstract path $\psi = o_0 \sigma_0 \dots o_n$ and
a function $f_0 \in \calB$ such that $\supp(f_0) \subseteq o_0$, let
$\prop(\psi,f_0)$ denote the unique (pointwise) $\preceq_0$-minimal
function-action sequence $f_0 \sigma_0 \dots f_n \in \supp^{-1}(\psi)$. That is
to say, if $\prop(\psi,f_0) = f_0 \sigma_0 \dots f_n$, then for all $g_0
\sigma_0 \dots g_n \in \supp^{-1}(\psi)$ it holds that $f_i \preceq_0 g_i$ for
all $0 \le i \le n$. Observe the latter holds if and only if $f_{i+1}$ is a
proper $\sigma_i$-successor of $f_i$, for all $0 \le i < n$.

We extend $\supp(\cdot)$, $\supp^{-1}(\cdot)$, and $\prop(\cdot,\cdot)$ to
infinite sequences in the obvious way.

\paragraph*{The weighted unfolding of an MPG}
The \emph{reachability game associated with $G$}, \ie $\Gamma = ( \Pi,
\Sigma, f_I, \delta, \teve, \tadam )$, is formally defined as follows.
The initial state $f_I \in \calB$ is the function for which $q_I \mapsto 0$, and
$q \mapsto \bot$ for all $q \neq q_I$. The state-set $\Pi$ is the subset of
$\calF(G)$ where for all $f_0 \sigma_0 f_1 \dots f_n \in \Pi$ we have:
\begin{itemize}[nolistsep]
	\item $f_0 = f_I$; and
	\item for all $0 \leq i < j <n$,
		\begin{itemize}[nolistsep]
			\item $f_i \not \preceq_0 f_j$	and
			\item $f_j \not \preceq_1 f_i$.
		\end{itemize}
\end{itemize}
The transition function $\delta$ is such that if $x$ and $x \cdot \sigma \cdot
f$ are elements of $\Pi$ then $(x,\sigma,x \cdot \sigma \cdot f) \in \delta$.
For the terminating states we have
\begin{align*}
	\teve = \{ f_0 \sigma_0 \dots f_n \in \Pi \st
	& \text{ for some }
	0 \leq i < n : f_i \preceq_0 f_n\}; \text{ and }\\
	\tadam = \{ f_0 \sigma_0 \dots f_n \in \Pi \st
	& \text{ for some } 0 \leq i < n : f_n \preceq_1 f_i, \text{ and } \\
	& \text{ for some } q \in \supp(f_i) : f_i(q) \neq+\infty \}.
\end{align*}
Note that the directed graph defined by $\Pi$ and $\delta$ is a tree, but not
necessarily finite.

\paragraph*{Good and bad cycles}
To gain an intuition about $\Gamma$, let us say an
abstract cycle $\chi$ is \emph{good} if:
\begin{itemize}[nolistsep]
	\item there exists $f_0 \sigma_0 \dots f_n \in \supp^{-1}(\chi)$ such
		that $f_i(q) \neq +\infty$ for all $q \in Q$ and all $0 \le i
		\le n$, and
	\item $f_0 \preceq_0 f_n$.
\end{itemize}
Let us say $\chi$ is \emph{bad} if:
\begin{itemize}[nolistsep]
	\item there exists $f_0 \sigma_0 \dots f_n \in \supp^{-1}(\chi)$ such
		that $f_0(q) \neq +\infty$ for some $q \in \supp(f_0)$, and
	\item $f_n \preceq_1 f_0$.
\end{itemize}
Then it is not difficult to see that $\Gamma$ is essentially an abstract
cycle-forming game played on $G$ which is winning for Eve if a good abstract
cycle is formed and winning for Adam if a bad abstract cycle is formed.

Our main result for this section is the following:
\begin{theorem}
\label{thm:stratTransfer}
	Let $G$ be an MPG with limited observation and let $\Gamma$ be the
	associated reachability game. If Adam (Eve) has a winning strategy in
	$\Gamma$ then (s)he has a finite-memory observation-based winning
	strategy in $G$.
\end{theorem}

The idea behind the observation-based strategies
for the MPG is straightforward. If Eve
wins the reachability game then she can transform her strategy into one that
plays indefinitely by returning, whenever the play reaches $\teve$, to the
natural previous position---namely the position that witnesses the membership
of $\teve$.  By continually playing her winning strategy in this way Eve
perpetually completes good abstract cycles and this ensures that all concrete
paths consistent with the play have non-negative mean-payoff value. Likewise if
Adam has a winning strategy in the reachability game, he can continually play
his strategy by returning to the natural position whenever the play reaches
$\tadam$. By doing this he perpetually completes bad abstract cycles and this
ensures that there is a concrete path consistent with the play that has strictly
negative mean-payoff value.

We will repeatedly use the next result which follows by induction immediately
from the definition of a $\sigma$-successor.

\begin{lemma}
	\label{lem:minPath}
	Let $\phi = f_0 \sigma_0 \dots f_n \in \calF(G)$ be a sequence
	such that $f_{i+1}$ is a proper $\sigma_i$-successor of $f_i$, for all
	$i$. Then for all $q \in \supp(f_n)$,
	\[
		f_n(q) = \min \{ f_0(\pi[0]) + w(\pi) \st \pi \in
			\gamma(\supp(\phi))\text{ and } \pi[n] = q\}.
	\]
\end{lemma}
The following simple facts about $\preceq_n$ will also be useful:
\begin{lemma}
\label{lem:wqo-trans}
For any two functions $f_1,f_2 \in \calB$ with $f_1 \preceq_k f_2$:
\begin{enumerate}[label=(\roman*),nolistsep]
	\item For all $k' \leq k$, $f_1 \preceq_{k'} f_2$,
	\item For all $k' \geq 0$, if $f_2 \preceq_{k'} f_3$ for some $f_3 \in
		\calF$ then $f_1 \preceq_{k+k'} f_3$, and
	\item If $f_1'$ is a proper $\sigma$-successor of $f_1$ and $f_2'$ is a
		$\sigma$-successor of $f_2$ with $\supp(f_2') = \supp(f_1')$, then
		$f_1' \preceq_k f_2'$.
\end{enumerate}
\end{lemma}
\begin{proof}
	Items (i) and (ii) are trivial. For (iii), let $d_{i,j} = w(q_i,\sigma,q_j)$
	for $q_i \in \supp(f_1)$ and $q_j \in \supp(f_1')$ where such a
	transition exists and $+\infty$ otherwise. We now observe that as $f_1'$
	is $\preceq_0$-minimal, $f_1'(q_j)$ can be defined as $\min\{f_1(q_i) +
	d_{i,j} \st q_i \in \supp(f_1)\}$ for all $q_j \in \supp(f_1')$.
	As $f_1(q_i) \le f_2(q_i) - k$ for any $q_i \in \supp(f_1)$, it
	follows that \[ f_1'(q_j) \le \min\{f_2(q_i) + d_{i,j} \st q_i
	\in \supp(f_1)\} - k \le f_2'(q_j) - k,\] where the second
	inequality follows from the definition of a $\sigma$-successor.
	Thus $f_1' \preceq_k f_2'$.
\end{proof}

Although the following results are not used until
Section~\ref{sec:adeqpure-inc}, they already give an intuition towards the
correctness of the strategies described above. In words, we will show that
repeating good cycles is itself, in some sense, good, while repeating bad ones
is bad.
\begin{lemma}\label{lem:caVa}
	Let $\chi$ be an abstract cycle.
	\begin{enumerate}[label=(\roman*),nolistsep]
		\item If $\chi$ is good (bad) then an interleaving of $\chi$
			with another good (bad) cycle is also good (bad).
		\item If $\chi$ is good then for all $k$ and all concrete cycles
			$\pi \in \gamma(\chi^k)$, $w(\pi) \geq 0$.
		\item If $\chi$ is bad then $\exists k \ge 0, \pi \in
			\gamma(\chi^k)$ such that $w(\pi)<0$.
	\end{enumerate}
\end{lemma}
\begin{proof}
	Item (i) follows from Lemma~\ref{lem:wqo-trans}. For (ii), let $f_0 \sigma_0
	\dots f_n \in \supp^{-1}(\chi)$ be such that $f_i(q) \neq +\infty$ for
	all $i$ and $q$ and $f_0 \preceq_0 f_n$. In particular this means that
	$f_{i+1}$ is a proper $\sigma_i$-successor of $f_i$.  Now fix $k$ and
	let $\pi \in \gamma(\chi^k)$ be a concrete cycle. From
	Lemma~\ref{lem:minPath} we have, for all $0 \leq i < k$,
	\[
		w(\pi[n\cdot i..n(i+1)]) \geq f_n(\pi[n(i+1)]) - f_0(\pi[n\cdot i])
	\]
	and
	\[
		f_n(\pi[n(i+1)]) - f_0(\pi[n\cdot i]) \geq f_0(\pi[n(i+1)]) -
		f_0(\pi[n\cdot i]).
	\]
	Hence
	\[
		w(\pi) = \sum_{i=1}^k w(\pi[n\cdot i..n(i+1)]) \geq f_0(\pi[n
		\cdot k]) - f_0(\pi[0]) = 0.
	\]

	We now prove item (iii) holds.
	Let $f_0 \sigma_0 \dots f_n \in \supp^{-1}(\chi)$ and $q_0 \in
	\supp(f_0)$ be such that $f_0(q_0) \neq +\infty$ and $f_n \preceq_1
	f_0$. It follows that $f_n(q_0) < +\infty$. From the definition of a
	$\sigma$-successor, it follows that there exists $r \in \supp(f_{n-1})$
	such that $f_{n-1}(r)<+\infty$, and there is an edge from $r$ to $q_0$
	with weight $f_n(q_0) - f_{n-1}(r)$. Proceeding this way inductively we
	find there is a $q_1 \in \supp(f_0)$ with $f_0(q_1) < +\infty$ and a
	concrete path $\pi_0\in \gamma(\chi)$ from $q_1$ to $q_0$ with $w(\pi_0)
	= f_n(q_0) - f_0(q_1)$. As $f_0(q_1) < +\infty$ and $f_n \preceq_1 f_0$
	we have $f_n(q_1) \leq f_0(q_1) - 1 < +\infty$. Repeating the argument
	yields a sequence of states $q_0, q_1, \dots$ such that there is a
	concrete path $\pi_i \in \gamma(\chi)$ from $q_{i+1}$ to $q_i$ with
	\[
		w(\pi_i) = f_n(q_i) - f_0(q_{i+1}) \leq f_0(q_i) - f_0(q_{i+1}) - 1.
	\]
	As $Q$ is finite it follows that there exists $i<j$ such that $q_i =
	q_j$. Then the concrete path $\pi = \pi_j \cdot \pi_{j-1} \dots
	\pi_{i+1} \in \gamma(\chi^{j-i})$ is a concrete cycle with
	\[
		w(\pi) = \sum_{k=i+1}^j w(\pi_k) \leq f_0(q_i) - f_0(q_j) -
		(j-i) < 0.
	\]
\end{proof}

\begin{corollary} \label{cor:quelleSauce}
	No cyclic permutation of a good abstract cycle is bad.
\end{corollary}

\paragraph*{Restricting $\Gamma$ w.r.t. a strategy}
We note that, as a play prefix in $\Gamma$ is completely described by the last
state in the sequence, it suffices to consider positional strategies for both
players. Thus, when speaking of winning strategies for either player in
$\Gamma$, we will assume they are positional.

Let $\restriction{\Pi}{\lambda}$ denote the set of states from $\Pi$ that can be
reached via plays consistent with $\lambda$.  At this point, we can already show
that a winning strategy for either player in $\Gamma$ will reach a terminating
state in a bounded number of steps. This will later allow us to argue that the
strategies we construct for Eve or Adam in $G$ based on their strategies in
$\Gamma$ use finite memory.
\begin{lemma}
\label{lem:boundedPlays}
	If $\lambda$ is a winning strategy for Adam or Eve in $\Gamma$, then there
	exists $N \in \mathbb{N}$ such that for all plays $\pi$ consistent with
	$\lambda$ we have $|\pi| \leq N$.
\end{lemma}
\begin{proof}
	Suppose there is no bound on the size of $\restriction{\Pi}{\lambda}$.
	As $\Gamma$, is acyclic, it follows that $\restriction{\Pi}{\lambda}$
	contains infinitely many states. However, as $\Gamma$ is
	finitely-branching, it follows from K\"onig's lemma that there exists an
	infinite path in $\Gamma$.  As this play is not winning for either
	player and it is consistent with $\lambda$, this contradicts the fact
	that $\lambda$ is a winning strategy.
\end{proof}

\subsection{Strategy transfer for Eve}
Suppose Eve has a winning positional strategy $\lambda$ in $\Gamma$.  Let $M =
\restriction{\Pi}{\lambda}$ be the corresponding restriction of $\Pi$. From
Lemma~\ref{lem:boundedPlays}, $M$ is finite. We will define an
observation-based strategy $\lambda^*$ with memory $|M|$ for Eve in $G$. Given a
memory state $\mu = f_0 \sigma_0 \dots f_n \in M$ let
\[
	\mu' = \begin{cases}
		\text{the proper prefix } f_0 \sigma_0 \dots f_\ell
			\text{ of } \mu \text{ such that }
			f_\ell \preceq_0 f_n
			& \text{if } \mu \in \teve\\
		\mu & \text{otherwise.}
		\end{cases}
\]
The initial memory state is $\mu_0 := f_I$. Let us write $\mu' = f'_0 \sigma'_0
\dots f'_{m}$.  We define the output function $\alpha_o:M \times \Obs \to
\Sigma$ as $\alpha_o(\mu,o) = \lambda(\mu')$. Finally we define the update
function $\alpha_u:M \times \Obs \to M$ as $\alpha_u(\mu, o) = \mu' \cdot
\lambda(\mu') \cdot g$ where $g$ is the proper $\lambda(\mu')$-successor of
$f'_{m}$ with $\supp(g) = o$. Observe that we maintain the invariant that the
current observation is $\supp(f'_{m})$, consequently the $\Obs$ input to
$\alpha_o$ is not used.

We will show shortly that $\lambda^*$ is a winning observation-based
strategy for Eve in $G$. First, we require a result about plays consistent with
$\lambda^*$. We will argue that, by following $\lambda^*$ in $G$, Eve ensures
the belief functions from proper function-action sequences
induced by play prefixes consistent with it are
$\preceq_0$-smaller than her current memory state.
\begin{lemma}\label{lem:alwaysGood}
	Let $\psi = o_0 \sigma_0 \dots \in \Plays(G)$ and
	$\mu_0 \mu_1 \dots \in M^\omega$
	be such that $\sigma_i = \alpha_o(\mu_i,o_i)$ and
	$\mu_{i+1} = \alpha_u(\mu_i,o_i)$ for all $i \ge 0$. If we write
	$\mu_i = f_i^{(0)} \sigma_i^{(0)} \dots f_i^{(n_i)}$ and
	$\prop(\psi,f_I) = g_0 \sigma_0 \dots$, then $f^{(n_i)}_i \preceq_0 g_i$
	for all $i \ge 0$.
\end{lemma}
\begin{proof}
	We prove this by induction. For
	$i=0$ we have
	\(
		\mu_0 = f_I = g_0.
	\)
	Now suppose $f^{(n_i)}_i \preceq_0 g_i$. By definition of
	$\prop(\cdot,\cdot)$ we have that $g_{i+1}$ is the proper
	$\sigma_i$-successor of $g_i$ and $\supp(g_{i+1}) = o_{i+1}$. Assume
	first that $\mu_i \notin \teve$.
	Then
	\[
		\mu_{i+1} = \alpha_u( \mu_i,o_i)
		= \mu_i \cdot \sigma_i \cdot h,
	\]
	where $h$ is the proper $\sigma_i$-successor of
	$f^{(n_i)}_i$ with $\supp(h) = o_{i+1}$.  Then, by
	Lemma~\ref{lem:wqo-trans}~(iii)
	we have $f^{(n_{i+1})}_{i+1} = h \preceq_0 g_{i+1}$.

	Now assume $\mu_i \in \teve$, and let
	$\ell < n_i$ be the index
	such that $f^{(\ell)}_i \preceq_0 f^{(n_i)}_i$.
	Then
	\[
		\mu_{i+1} = \alpha_u(\mu_i,o_i) =
		\left(f_i^{(0)} \sigma_i^{(0)} \dots \sigma_{(i)}^{(\ell-1)}
		f_i^{(\ell)}\right) \cdot \sigma_i \cdot h
	\]
	where $h$
	is the proper $\sigma_i$-successor
	of $f^{(\ell)}_i$ with $\supp(h) = o_{i+1}$.  From
	Lemma~\ref{lem:wqo-trans} item (ii) we have $f^{(\ell)}_i \preceq_0
	g_i$, so by Lemma~\ref{lem:wqo-trans}~(iii) we have
	$f^{(n_{i+1})}_{i+1} = h \preceq_0 g_{i+1}$ as required.
\end{proof}

We now proceed with the proof of strategy transfer for Eve.
\begin{lemma}
\label{lem:EveIf}
	Let $G$ be a mean-payoff game with limited observation and let
	$\Gamma$ be the associated reachability game. If Eve has a winning
	strategy in $\Gamma$ then she has a finite-memory observation-based
	winning strategy in $G$.
\end{lemma}

\begin{proof}
	We will show that $\lambda^*$ described above is a winning strategy for
	Eve. Let $\psi = o_0 \sigma_0 \dots \in \Plays(G)$ be any play
	consistent with $\lambda^*$. That is, there is a sequence $\mu_0 \mu_1
	\dots M^\omega$ such that $\sigma_i = \alpha_o(\mu_i,o_i)$ and
	$\mu_{i+1} = \alpha_u(\mu_i,o_i)$ for all $i \ge 0$.  We will show that
	there exists a constant $\beta \geq 0$ such that for all concrete paths
	$\pi \in \gamma(\psi)$ and all $j\geq 0$, $w(\pi[..j]) \geq \beta$. It
	follows that $\uMP(\pi) \geq 0$, and so $\psi$ is winning for Eve.

	Let $W = \{f_\ell(q) \st f_0 \sigma_0 \dots f_\ell \in M, q \in Q,\text{
	and } f_\ell(q) \neq \bot\}.$ Note that $W$ is finite because $M$ and $Q$
	are finite, and non-empty because $\mu_0 = f_I$ and $f_I(q_I) = 0 \in
	W$. Let $\beta = \min W$. As $0 \in W$, we have that $\beta < +\infty$.

	As with Lemma~\ref{lem:alwaysGood}, let $\prop(\psi,f_I) = g_0 \sigma_0
	\dots$ and $\mu_i = f_i^{(0)} \sigma_i^{(0)} \dots f_i^{(n_i)}$.
	Consider an arbitrary $j \in \mathbb{N}$.
	As $\supp(f_I) = \{q_I\}$ and $f_I(q_I)=0$,
	Lemma~\ref{lem:minPath} implies for all $q \in
	\supp(g_j)$, we have $g_j(q) \neq +\infty$. Hence,
	for all concrete paths $\pi \in \gamma(\psi)$ we have:
	\[
		\begin{array}{rcll}
			w(\pi[..j]) & \geq & g_j(\pi[j]) -
			f_I(\pi[0]) & \text{from Lemma~\ref{lem:minPath}}\\
			& = & g_j(\pi[j]) \\
			& \geq & f_j^{(n_j)}(\pi[n]) &
				\text{from Lemma~\ref{lem:alwaysGood}}\\
			& \geq & \beta & \text{as required.}
		\end{array}
	\]
\end{proof}

\subsection{Strategy transfer for Adam}
To complete the proof of Theorem~\ref{thm:stratTransfer}, we now show how to
transfer a winning strategy for Adam from $\Gamma$ to a winning finite-memory
observation-based strategy in $G$.  So let us assume $\lambda:\Pi \times \Sigma
\to \Pi$ is a (positional) winning strategy for Adam in $\Gamma$. The
finite-memory observation-based strategy for Adam is similar to that for Eve in
that it perpetually plays $\lambda$, returning to a previous position whenever
the play reaches $\tadam$.  However, the proof of correctness is more intricate
because we need to handle the $+\infty$ function values.

Formally, the finite-memory observation-based
strategy $\lambda^*$ is given as follows.  As
before, let $M = \restriction{\Pi}{\lambda}$
and $\mu_0 = f_I$. Given $\mu \in M$, let
\[
	\mu' = \begin{cases}
		\text{the proper prefix of }
			f_0 \sigma_0 \dots f_\ell
			\text{ such that } f_\ell \preceq_1 f_n
			& \text{if } \mu \in \tadam\\
		\mu & \text{otherwise.}
	\end{cases}
\]
Let us write $\mu' = f'_0 \sigma'_0 \dots f'_m$.
The output function $\alpha_o:M \times \Obs \times \Sigma \to \Obs$ is defined as:
$\alpha_o(\mu,o,\sigma) = \supp(g)$ where $\lambda(\mu',\sigma)) =
\mu' \cdot \sigma \cdot g$. The update function $\alpha_u:M \times \Obs \times
\Sigma \to M$ is defined as: $\alpha_u(\mu, o, \sigma) = \lambda(\mu',\sigma)$.
Note that as the current observation is stored in the memory state, the $\Obs$
input to $\alpha_o$ and $\alpha_u$ is redundant.

To show that $\lambda^*$ is winning for Adam in $G$ we require an analogue to
Lemma~\ref{lem:alwaysGood}. To be precise, we show that, by following
$\lambda^*$ in $G$, Adam ensures the belief functions from \emph{ultimately
proper function-action sequences} induced by play prefixes consistent with it
are $\preceq_r$-larger than his current memory state (for $r$ a function of how
many times his memory has been \emph{reset}).  Given a finite sequence $\mu_0
\dots \mu_n \in M^*$ of memory states we denote by $\reset(\mu_0 \dots \mu_n)$
the number of times the memory is reset along the sequence. That is,
$\reset(\mu_0) = 0$, and if $\mu_i \in \tadam$ then
$\reset(\mu_0 \dots \mu_{i+1}) = \reset(\mu_0 \dots \mu_i) + 1$, otherwise
$\reset(\mu_0 \dots \mu_{i+1}) = \reset(\mu_0 \dots \mu_i)$.

\begin{lemma}\label{lem:alwaysBad}
	Let $\psi = o_0 \sigma_0 \dots \in \Plays(G)$, $\mu_0 \mu_1 \dots \in
	M^\omega$, $k \in \mathbb{N}$, and $f_0 \sigma_0 \dots \in
	\supp^{-1}(\psi)$ be such that:
	\begin{itemize}[nolistsep]
		\item $f_k \sigma_k \dots = \prop(\psi[k..],f_k)$; and for all
			$i \ge 0$,
		\item $o_{i+1} = \alpha_o(\mu_i,o_i,\sigma_i)$ and
		\item $\mu_{i+1} = \alpha_u(\mu_i,o_i,\sigma_i)$.
	\end{itemize}
	If we write $\mu_i = g_i^{(0)} \sigma_i^{(0)} \dots g_i^{(n_i)}$ and $f_k
	\preceq_r g_k^{(n_k)}$ for some $r \in \mathbb{N}$, then for all $i \ge k$
	it holds that
	\(
		f_i \preceq_{r'_i} g_i^{(n_i)}
	\)
	where $r'_i = r + \reset(\mu_k \dots \mu_i)$.
\end{lemma}
\begin{proof}
	We prove this by induction on $i$. For $i=k$ the result clearly holds.
	Now suppose $i \geq k$ and $f_i \preceq_{r'_i} g_i^{(n_i)}$ where $r'_i =
	r + \reset(\mu_k \dots \mu_i)$. We consider two cases depending on
	whether $\mu_i \in \tadam$. If $\mu_{i} \notin \tadam$ then
	\[
		\mu_{i+1} = \mu_{i} \cdot \sigma_i \cdot h
	\]
	where $h$ is a $\sigma_i$-successor of
	$g_i^{(n_i)}$ with $\supp(h) = o_{i+1}$. Furthermore, since no reset
	occurred, we have that $\reset(\mu_0 \dots
	\mu_{i+1}) = \reset(\mu_0 \dots \mu_i)$. Then, by
	Lemma~\ref{lem:wqo-trans}~(iii) we have $f_{i+1} \preceq_{r'_i} h =
	g_{i+1}^{(n_{i+1})}$ and since
	\[
		r'_i = r + \reset(\mu_0 \dots \mu_i) = r + \reset(\mu_0 \dots
		\mu_{i+1}) = r'_{i+1}
	\]
	the result holds for $i+1$.

	Otherwise if $\mu_i \in \tadam$, let $\ell < n_i$ be the index such that
	$g_i^{(n_i)} \preceq_1 g_i^\ell$. We have that
	\[
		\mu_{i+1} = \left(g_i^{(0)} \sigma_i^{(0)} \dots
		\sigma_i^{(\ell-1)}g_i^{(\ell)} \right) \cdot \sigma_{i} \cdot h
	\]
	where $h$ is a $\sigma_{i}$-successor of $g_i^{(\ell)}$ with $\supp(h) =
	o_{i+1}$.  From Lemma~\ref{lem:wqo-trans}~(ii) we have $f_i
	\preceq_{r'_i+1} h$. So by Lemma~\ref{lem:wqo-trans}~(iii) we have
	$f_{i+1} \preceq_{r'_i + 1} h = g_{i+1}^{(n_{i+1})}$, and as
	\[
		r'_i + 1 = r + \reset(\mu_k \dots \mu_i) + 1 =
		r + \reset(\mu_k \dots \mu_{i+1}) =
		r'_{i+1}
	\]
	the result holds for $i+1$.
\end{proof}

We now show how to transfer strategies for Adam.
\begin{lemma}
	\label{lem:Adam}
	Let $G$ be a mean-payoff game with limited observation and let
	$\Gamma$ be the associated reachability game. If Adam has a winning
	strategy in $\Gamma$ then he has a finite-memory observation-based
	winning strategy in $G$.
\end{lemma}

\begin{proof}
	We will show that the finite-memory observation-based
	strategy $\lambda^*$ constructed above is winning
	for Adam. Let $\psi = o_0 \sigma_0 \dots$ be any play consistent with
	$\lambda^*$. That is, there is a sequence $\mu_0 \mu_1 \dots M^\omega$
	such that $o_{i+1} = \alpha_o(\mu_i,o_i,\sigma_i)$ and $\mu_{i+1} =
	\alpha_u(\mu_i,o_i,\sigma_i)$ for all $i \ge 0$. Let us write $\mu_i =
	g^{(0)}_i \sigma_i^{(0)} \dots g_i^{(n_i)}$.
	As $M$ is finite, there exists $\phi = \dots f^* \in M$ and an infinite
	set $\calI\subseteq \mathbb{N}$ of indices such that for all $i \in
	\calI$ we have $\mu_{i} = \phi$.  We will show that this implies there
	exists $\pi \in \gamma(\psi)$ such that $\oMP(\pi)<0$. As $\oMP(\pi)
	\geq \uMP(\pi)$ the result follows.  For convenience, given $n \in
	\mathbb{N}$, let $\suc{n} = \min \{ i \in \mathcal{I} \st i > n\}$.
	Denote by $o^*$ the set $\{q \in \supp(f^*) \st f^*(q) \neq +\infty\}$.
	Note that from the definition of $\tadam$ it follows that
	$o^*$ is non-empty.

	We will use a function-action sequence to find a concrete path that is
	winning for Adam. That is, a concrete path where the
	weights of the prefixes can be identified and seen to be strictly
	decreasing. Unlike in Lemma~\ref{lem:EveIf}, the unique proper
	function-action sequence $\prop(\psi,f_I)$ fulfilled does not hold
	enough information for us to prove this claim. Indeed, to handle
	$+\infty$ values, which correspond to irrelevant paths, we require a
	more complex sequence.  Recall that Adam can place $+\infty$ values in a
	function to tell Eve that the token is \textbf{not} in a particular
	state, \eg $f(q) = +\infty$ signifies $q$ does not hold the token.

	The sequence we construct will be \emph{piecewise proper} in the sense
	that for all $i \in \calI$ the sequence will consist of proper
	successors in the interval $[i,\suc{i})$.  When the sequence reaches an
	element of $\calI$ we ``reset'' the values of the states not in
	$o^*$ to $+\infty$.  More formally, the required sequence,
	$\pprop(\psi,f_I) = h_0 \sigma_0 \dots \in \supp^{-1}(\psi)$, is constructed
	inductively as follows.  Initially, let $h_0 = h_0' = f_I$. For $i \geq
	0$, let $h_{i+1}'$ be the proper $\sigma_i$-successor of $h_i$ with
	$\supp(h'_{i+1}) = o_{i+1}$. If $i \notin \calI$ then $h_i = h_i'$.
	Otherwise, for any $q \in \supp(h'_{i+1})$ we let
	\[
		h'_{i}(q) = \begin{cases}
				+\infty & \text{if } q \notin o\\
				h'_{i}(q) & \text{otherwise.}
			\end{cases}
	\]
	Observe that, by construction, for all $i \in \calI$ and all $q \in
	o^*$ we have that $h_i(q) = h'_i(q)$.

	We now claim that
	\[
		\forall i \in \mathbb{N} :
		h_i \preceq_{r_i} g_i^{(n_i)},
	\]
	where $r_i = \reset(\mu_0 \dots \mu_i)$. From
	Lemma~\ref{lem:wqo-trans}~(i) and Lemma~\ref{lem:alwaysBad} it follows
	that we only need to show that for all $i \in \mathcal{I}$ it holds that
	\(
		h_{i} \preceq_{r_i} f^*.
	\)
	Induction and Lemma~\ref{lem:alwaysBad} imply that for all $i \in \calI$
	we have $h_{i}' \preceq_{r_i} f^*$. Recall that $h_{i}$ differs
	from $h'_{i}$ only on states where $f^*$ is equal to $+\infty$,
	we therefore have $h_{i} \preceq_{r_i} f^*$ as required.

	We will now show that there is an infinite concrete path $q_0 \sigma_0
	\dots \in \gamma(\psi)$ such that $q_{i} \in o^*$ for all $i \in
	\calI$. To do this we will show for any $i \in \calI$ and any $q \in
	o^*$ there is a concrete path in $\gamma(\psi[i..\suc{i}])$, that
	ends in $q$ and starts at some state in $o^*$. The result then follows by
	induction. Let us fix $i\in \calI$, $q \in o^*$, and let $j = \suc{i}$.
	As $h_{j}' \preceq_{r_{j}} f^*$, we have that $h'_{j}(q)
	\neq +\infty$.  From Lemma~\ref{lem:minPath}, there is a concrete path
	$\pi=q_0\sigma_0\dots q_n$ from $q_0 \in o_{i}$ ending at $q_n=q$ such
	that $h'_{j}(q) = h_{i}(q_0) + w(\pi)$. As $h'_{j}(q) \neq +\infty$ it
	follows that $h_{i}(q_0) \neq +\infty$, and as $h_{i}(q_0) = +\infty$ if
	and only if $f^*(q_0) = +\infty$, it follows that $q_0 \in o^*$.
	Note that Lemma~\ref{lem:minPath} implies for all $k\leq |\pi|$ we have
	\[
		w(\pi[..k]) = h'_{i+k}(q_k) - h_{i}(q_0) = h_{i+k}(q_k) - h_{i}(q_0).
	\]
	In particular $w(\pi) = h_{j}(q) - h_{i}(q_0)$.

	Now let $\pi = q_0 \sigma_0 \dots$ be the infinite path implied by the
	above construction and, for convenience, for $i \in \calI$ let $\pi_{i}
	= q_{i} \sigma_{i} \dots q_{j}$ where $j = \suc{i}$.  To show
	$\oMP(\pi)<0$ we need to show $\limsup_{\ell \rightarrow \infty}
	\frac{1}{\ell} w(\pi[..\ell])<0$.  To prove this, we will show there exists a
	constant $\beta<0$ such that for all sufficiently large $\ell$ we have
	$w(\pi[..\ell]) \leq \beta \cdot \ell$.

	For convenience, let $i_0 = \min\calI$ and let $i_\ell = \max\{ i \in
	\mathcal{I} \st i \leq \ell\}$.  From Lemma~\ref{lem:minPath} and
	the construction of $\pi_i$ we have for all $\ell \ge i_0$:
	\begin{align*}
		w(\pi[..\ell]) & = w(\pi[..i_0]) + w(\pi[i_\ell..\ell]) +
			\sum_{\substack{i \in \mathcal{I}\\i \leq \ell}} w(\pi_i)\\
		& = w(\pi[..i_0]) + (h_\ell(q_\ell) - h_{i_\ell}(q_{i_\ell}))
			+ \sum_{\substack{i \in \mathcal{I}\\i \leq \ell}}
			h_{\suc{i}}(q_{\suc{i}}) - h_{i}(q_{i})\\
		& = w(\pi[..i_0]) + h_{\ell}(q_{\ell}) - h_{i_0}(q_{i_0})\\
		& \leq w(\pi[..i_0]) + g_\ell^{(n_\ell)}(q_\ell) - r_\ell -
			h_{i_0}(q_{i_0})\\
	\end{align*}
	There are only finitely many values for $g_\ell^{(n_\ell)}(q_\ell)$ and
	from Lemma~\ref{lem:boundedPlays} we get $r_\ell \geq \lfloor \frac{\ell}{N}
	\rfloor$. Hence \[
		w(\pi[..\ell]) \leq \alpha - \beta'\cdot \ell
	\]
	for constants $\alpha$ and $\beta'>0$. Thus there exists $\beta < 0$
	such that for sufficiently large $\ell$ we have $w(\pi[..\ell]) \leq
	\beta \cdot \ell$.  Hence $\uMP(\pi) \leq \oMP(\pi) <0$.
\end{proof}

\section{Forcibly Terminating Games}\label{sec:ATgames}
The reachability game defined in the previous section gives a sufficient
condition for determining the winner in an MPG with limited observation.
However, as there may be plays where no player wins, such games are not
necessarily determined. The first subclass of MPGs with limited observation we
investigate is the class of games where the associated reachability game is
determined.
\begin{definition}
	An MPG with limited observation is \emph{forcibly terminating} if in
	the corresponding reachability game $\Gamma$ either Adam has a winning
	strategy to reach states in $\tadam$ or Eve has a winning strategy to
	reach states in $\teve$.
\end{definition}
It follows immediately from Theorem~\ref{thm:stratTransfer} that finite-memory
strategies suffice for both players in forcibly terminating games. Note that
an upper bound on the memory required is the number of states in the reachability
game restricted to a winning strategy, and this is exponential in $N$, the bound
obtained in Lemma~\ref{lem:boundedPlays}.

\begin{theorem}[Finite-memory determinacy]\label{thm:fin-mem-det}
	One player always has a winning finite-memory observation-based strategy
	in a forcibly terminating MPG.
\end{theorem}

We now consider the complexity of two natural decision problems associated with
forcibly terminating games: the problem of recognizing if an MPG is forcibly
terminating and the problem of determining the winner of a forcibly terminating
game. Both results follow directly from the fact that we can accurately simulate
a Turing Machine with an MPG with limited observation.

\begin{proposition}\label{pro:superThm}
	Let $M$ be a Deterministic Turing Machine. Then there exists an MPG with
	limited observation $G$, constructible in polynomial time, such that Eve
	wins the associated reachability game $\Gamma$ if and only if $M$ halts
	in the accept state and Adam wins $\Gamma$ if and only if $M$ halts in
	the reject state.
\end{proposition}
We will in fact show how to simulate a (deterministic) four-counter
machine ($4$CM). The standard reduction from Turing Machines to $4$CMs, via
finite state machines with two stacks (see \eg~\cite{minsky67}), is
readily seen to be constructible in polynomial time.

\paragraph*{Counter machines}
A counter machine (CM) $M$ consists of a finite set of control states
$S$, an initial state $s_I \in S$, a final accepting state $s_A \in S$,
a final rejecting state $s_R$, a set $C$ of integer-valued counters and
a finite set $\delta_M$ of instructions manipulating the counters.
$\delta_M$ contains tuples $(s,instr,c,s')$ where $s,s' \in S$ are
source and target states respectively, the action $instr \in \{\mathrm{INC},
\mathrm{DEC}\}$ applies to counter $c \in C$.  It also contains tuples of the
form $(s,0\mathrm{CHK},c,s',s^0)$ where $s',s^0$ are two target states, one of
which will be chosen depending on the value of counter $c$ at the moment
the instruction is ``executed''. Without loss of generality we may
assume $M$ is deterministic in the sense that for every state $s \in S$
there is exactly one instruction of the form
$(s,0\mathrm{CHK},\cdot,\cdot,\cdot)$ in $\delta_M$ or one of the form
$(s,\cdot,\cdot,\cdot)$. We also assume that $\mathrm{DEC}$ instructions are
always preceded by $0\mathrm{CHK}$ instructions so that counter values never go
below $0$.

A \emph{configuration} of $M$ is a pair $(s,v)$ where $s \in S$ and $v : C
\to \mathbb{N}$ is a valuation of the counters. A \emph{valid run} of
$M$ is a finite sequence $(s_0,v_0) \delta_0 (s_1,v_1) \delta_1 \dots$
$\delta_{n-1} (s_n,v_n)$ where $\delta_i \in \delta_M$ is either
$(s_i,instr_i,c_i,s'_i)$ or $(s_i,instr_i,c_i,s'_i,s^0_i)$ and
$(s_i,v_i)$ are configurations of $M$ such that $s_0 = s_I$, $v_0(c) = 0$ for
all $c \in C$, and for all $0 \le i < n$ we have that:
\begin{itemize}[nolistsep]
	\item $v_{i+1}(c) = v_i(c)$ for $c \in C \setminus \{c_i\}$;
	\item if $instr_i = \mathrm{INC}$ then $v_{i+1}(c_i) = v_i(c_i) + 1$ and
		$s_{i+1} = s'_i$;
	\item if $instr_i = \mathrm{DEC}$ then $v_{i+1}(c_i) = v_i(c_i) - 1$ and
		$s_{i+1} = s'_i$;
	\item if $instr_i = 0\mathrm{CHK}$ then $v_{i+1}(c_i) = v_i(c_i)$ and if
		$v_i(c_i) = 0$ we have $s_{i+1} = s^0_i$, otherwise
		$s_{i+1} = s'_i$.
\end{itemize}
The run is \emph{accepting} if $s_n = s_A$ and it is \emph{rejecting} if
$s_n = s_R$.

\begin{proof}[Proof of Proposition~\ref{pro:superThm}]
	Given a $4$CM $M = ( S,s_I,s_A,s_R,C,\delta_M )$, we now
	show how to construct an MPG with limited observation $G$ in which Eve
	wins the associated reachability game $\Gamma$ if and only if $M$ has an
	accepting run, and Adam wins $\Gamma$ if and only if $M$ has a rejecting
	run.  Plays in $G$ correspond to executions of $M$. As we will see, the
	tricky part is to make sure that zero-check instructions are faithfully
	simulated by one of the players.  Initially, both players will be
	allowed to declare how many instructions the machine needs to execute in
	order to reach an accepting or rejecting state. Either player can bail
	out of this initial ``pumping phase'' and become the \emph{Simulator}.
	The Simulator is then responsible for the faithful simulation of $M$ and
	the opponent will be monitoring the simulation and punish him if the
	simulation is not executed correctly.  Let us now go into the details.

	\item \paragraph*{Control structure}
	The control structure of the machine $M$ is encoded in the
	observations of our game. To be precise, to each state of $M$, there
	will correspond at most three observations in the game. We require two
	copies of each such observation since, in order to punish Adam or Eve
	(whoever plays the role of Simulator), existential and universal gadgets
	have to be set up in a different manner. For technical reasons that will
	be made clear below, we also need two additional observations. Formally,
	the observation set in our game contain observations $\{b^+,b^0,b^-\}$,
	$\{a^+,a^-\}$, and $\{q_I\}$, which do not correspond to instructions
	from the $4$CM but they are used in gadgets that will make sure that zero
	tests are faithfully executed.

	\item \paragraph*{Counter values}
	The values of counters will be encoded using the weights of
	transitions that reach designated states. We will associate to each
	observation (so to each state in the $4$CM) two states for each
	counter: $c_i^+$ and $c_i^-$, for $i \in \{1,2,3,4\}$. Intuitively, an
	abstract path, corresponding to the simulation of a run of the machine,
	will encode the value of counter $i$, at each step, as the weight of the
	shortest suffix from the initial pumping gadget to $c_i^+$.

	\item \paragraph*{Start of the construction}
	The mean-payoff game with limited observation $G = ( Q,q_I,\Sigma,\Delta,w,\Obs )$
	starts in observation $\{q_I\}$. For now, we will describe the
	transitions of $G$ on symbols $\Sigma' \subset \Sigma$ which allow Eve
	to ``declare'' the value of a counter as being zero or non-zero, as well as
	``bailing'' from certain gadgets in the game. More formally, let
	$\Sigma' := \{z,\overline{z},\texttt{bail}\}$. The transition relation
	$\Delta$ contains $\sigma$-transitions (for all $\sigma \in \Sigma'$)
	from $q_I$ to $b^+,b^0,b^-$. This observation represents the pumping
	phase of the simulation. From here each player will be allowed to
	declare how many steps they require to reach a halting state that will
	accept or reject. If Adam bails, we go to the initial instruction of $M$
	on the universal side of the construction ($s^\forall_I$), if Eve does
	so then we go to the analogue in the existential side ($s^\exists_I$).
	$\Sigma'$ contains a symbol $\texttt{bail}$ which represents Eve
	choosing to leave the gadget and
	try simulating an accepting run of $M$, that is $\Delta \ni
	(b^+,\texttt{bail},(s^\exists_I,\alpha^-))$,
	$(b^-,\texttt{bail},(s^\exists_I,\alpha^+))$, and
	$(b^0,\texttt{bail},(s^\exists_I,c))$ for $c \in \{c_i^+,c_i^- \st i =
	1,\ldots, 4\}$. For all other actions in $\Sigma'$, self-loops are added on
	the states $b^+,b^0,b^-$ with weights $+1,0,-1$ respectively. Meanwhile,
	Adam is able to exit the gadget at any moment---via non-deterministic
	transitions $(b^+,\sigma,(s^\forall_I,\alpha^-))$,
	$(b^-,\sigma,(s^\forall_I,\alpha^+))$, $(b^0,\sigma,(s^\forall_I,c))$
	where $c \in \{c_i^+,c_i^-\}$ for all $i$ and $\sigma \in \Sigma'
	\setminus \{\texttt{bail}\}$---to the universal side of the
	construction, \ie he will try to simulate a rejecting run of the
	machine. Bailing transitions (transitions going to states
	$(s^\exists_I,\cdot)$ or $(s^\forall_I,\cdot)$) have weight $0$.

\begin{figure}
\begin{center}
\begin{tikzpicture}
	\node[state](A){$c_i^+$};
	\node[state](B)[right=of A]{$c_i^-$};
	\node[state](C)[right=of B]{$\alpha^+$};
	\node[state](D)[right=of C]{$\alpha^-$};
	\node[fit=(A) (B) (C) (D), label=right:$s^\exists_I$]{};

	\node[state,below= 1.3cm of A](A2){$b^0$};
	\node[state,right= 2.5cm of A2](B2){$b^-$};
	\node[state,below= 1.3cm of D](C2){$b^+$};
	\node[fit=(A2) (B2) (C2)]{};

	\node[state,below= 1.3cm of A2](A3){$c_i^+$};
	\node[state,right= of A3](B3){$c_i^-$};
	\node[state,right=of B3](C3){$\alpha^+$};
	\node[state,right=of C3](D3){$\alpha^-$};
	\node[fit=(A3) (B3) (C3) (D3),label=right:$s^\forall_I$]{};

	\path
	(A2) edge[loopleft] node[el]{$\Sigma\setminus\{\texttt{bail}\},0$} (A2)
	(B2) edge[loopleft] node[el]{$\Sigma\setminus\{\texttt{bail}\},-1$} (B2)
	(C2) edge[loopleft] node[el]{$\Sigma\setminus\{\texttt{bail}\},1$} (C2)

	(A2) edge node[el,swap]{$\Sigma \setminus \{\texttt{bail}\},0$} (A3)
	(A2) edge node[el,pos=0.7]{$\Sigma \setminus \{\texttt{bail}\},0$} (B3)
	(B2) edge node[el,pos=0.7]{$\Sigma \setminus \{\texttt{bail}\},0$} (C3)
	(C2) edge node[el]{$\Sigma \setminus \{\texttt{bail}\},0$} (D3)

	(A2) edge node[el]{$\texttt{bail},0$} (A)
	(A2) edge node[el,swap,pos=0.7]{$\texttt{bail},0$} (B)
	(B2) edge node[el,swap,pos=0.7]{$\texttt{bail},0$} (C)
	(C2) edge node[el,swap]{$\texttt{bail},0$} (D);
\end{tikzpicture}
\end{center}
\caption{Initial pumping gadget for the $4$CM simulation}
\label{fig:pump-gadget}
\end{figure}
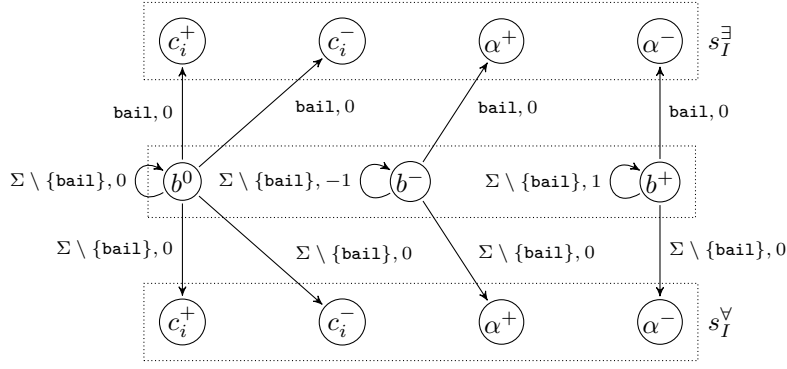

	Note that after these initial transitions the (simulated) value of all
	the counters is $0$.  Indeed, this corresponds to the beginning of a
	simulation of $M$ starting from configuration $(s_I, v)$ where $v(c) =
	0$ for all $c \in C$.

	\item \paragraph*{Counter increments \& decrements}
	Let us now explain how Eve simulates increments of counter values using
	this encoding (decrements are treated similarly). The gadget we explain
	below actually works the same in both sides of the construction, \ie
	the universal and existential gadgets for increments and decrements are
	identical. For that, consider Figure~\ref{fig:inc-gadget2}, the upper
	part of the figure is related to the state $s$ of $M$,
	while the bottom part is related to the state $s'$ of $M$, and assume
	that $(s,\mathrm{INC},c_i,s') \in \delta_M$.

	As can be seen in the figure, the observation related to the instruction
	$s$ contains the states $c_i^+,c_i^-$. These states are used for the
	encoding of the value of counter $c_i$.
	The additional states $\alpha^+,\alpha^-$ are used
	to encode the number of steps in the simulation (again one positive
	ending in $\alpha^+$ and one negative encoding in $\alpha^-$). Now, let
	us consider the transitions of the gadget. The increment of the counter
	$c_i$ from state $s$ to state $s'$ is encoded using the weights on
	the transitions that go from the observation $s$ to the observation
	$s'$. As you can see, the weight on the edge between the copy of state
	$c_i^+$ of observation $s$ to the copy of this state in observation $s'$
	is equal to $+1$, while the weight on the edge between the copy of state
	$c_i^-$ of observation $s$ to the copy of this state in observation $s'$
	is equal to $-1$.
	As you can see from the figure, when going from state $s$ to state
	$s'$, we also increment the additional counter that keeps track of the
	number of steps in the simulation of $M$. As the machine is
	deterministic there is no choice for Eve in observation $s$, since only
	an increment can be executed, this is why, regardless of the action
	chosen from $\Sigma'$, the same transition is taken.

\begin{figure}
\begin{minipage}[b]{0.49\linewidth}
\begin{center}
\begin{tikzpicture}[node distance=0.5cm]
	\node[state](A){$c_i^+$};
	\node[state](B)[right=of A]{$c_i^-$};
	\node[state](C)[right=of B]{$\alpha^+$};
	\node[state](D)[right=of C]{$\alpha^-$};
	\node[fit=(A) (B) (C) (D), label=right:$s$]{};

	\node[state,below= 1.3cm of A](A2){$c_i^+$};
	\node[state,right= of A2](B2){$c_i^-$};
	\node[state,right=of B2](C2){$\alpha^+$};
	\node[state,right=of C2](D2){$\alpha^-$};
	\node[fit=(A2) (B2) (C2) (D2),label=right:$s'$]{};

	\path
	(A) edge node[el]{$\Sigma',1$} (A2)
	(B) edge node[el]{$\Sigma',-1$} (B2)
	(C) edge node[el]{$\Sigma',1$} (C2)
	(D) edge node[el]{$\Sigma',-1$} (D2);
\end{tikzpicture}
\end{center}
\caption{Observation gadget for $(s,\mathrm{INC},c_i,s')$ instruction. For $(s,q) \in Q$
	only the $q$ component is shown.}
\label{fig:inc-gadget2}
\end{minipage}
\hfill
\begin{minipage}[b]{0.49\linewidth}
\begin{center}
\begin{tikzpicture}[node distance=0.5cm]
	\node[state](A){$c_i^+$};
	\node[state](B)[right=of A]{$c_i^-$};
	\node[state](C)[right=of B]{$\alpha^+$};
	\node[state](D)[right=of C]{$\alpha^-$};
	\node[fit=(A) (B) (C) (D),label=right:$s$]{};

	\node[state](I)[left=of A,initial above]{$q_I$};
	\node[fit=(I)]{};

	\node[state,below=1.3cm of A](A2){$c_i^+$};
	\node[state,right=of A2](B2){$c_i^-$};
	\node[state,right=of B2](C2){$\alpha^+$};
	\node[state,right=of C2](D2){$\alpha^-$};
	\node[fit=(A2) (B2) (C2) (D2),label=right:$s^0$]{};

	\path[]
	(A) edge node[el]{$z,0$} (A2)
	(B) edge node[el]{$z,0$} (B2)
	(C) edge node[el]{$z,1$} (C2)
	(D) edge node[el]{$z,-1$} (D2);

	\path[]
	(A) edge[out=-120,in=-60,looseness=1] node[el,pos=0.45]{$\overline{z},-1$} (I)
	(B) edge[out=150,in=30,looseness=1] node[el,swap,pos=0.6]{$z,0$} (I);
\end{tikzpicture}
\end{center}
\caption{Existential observation gadget for $(s,0\mathrm{CHK},c_i,s',s^0)$ instruction.
	 Transitions to $s'$ observation not shown.}
\label{fig:0chk-gadget2}
\end{minipage}
\end{figure}

	\item \paragraph*{Existential zero checks}
	Now, let us turn to the gadget of Figure~\ref{fig:0chk-gadget2}, that is
	used to simulate zero-check instructions. We first focus on the case
	in which it is the duty of Eve to reveal if the counter has value zero
	or not, by forcibly choosing the next letter to play in $\{z,
	\overline{z}\} \subset \Sigma'$. In the observation that
	corresponds to the state $s$ of $M$, Eve decides to declare
	that the counter $c_i$ is equal to zero (by issuing $z$) or not
	(by issuing $\overline{z}$), then Adam resolves non-determinism
	as follows. If Eve does not cheat then Adam should let the
	simulation continue to either $s^0$ or $s'$ depending on
	Eve's choice (the figure only depicts the branching to $s^0$,
	the branching to $s'$ is similar). Now if Eve has cheated, then
	Adam should have a way to retaliate: we allow him to do so by
	branching to observation $\{q_I\}$ from state $(s,c^-_i)$ with
	weight $0$ in case $z$ has been issued and the counter $c_i$ is
	not equal to zero and with weight $-1$ in case $\overline{z}$
	has been issued and the counter $c_i$ is equal to zero. It
	should be clear that in both cases Adam closes a bad abstract
	cycle.

	\item \paragraph*{Universal zero checks}
	A similar trick is used for the gadget from
	Figure~\ref{fig:0chk-gadget-adam}, where Adam is forced to simulate a
	truthful zero check or lose $\Gamma$. Since Adam can control
	non-determinism and not the action chosen, we have transitions going
	from $(s,\cdot)$ to states in both $(s_z,\cdot)$ and
	$(s_{\overline{z}})$ with weight $0$ and all actions in $\Sigma'$. Eve
	is then allowed to branch back to $q_I$ as follows. If Adam does not
	cheat, then Eve will play any action in $\Sigma' \setminus
	\{\texttt{bail}\}$ and transitions, with weights similar to those used
	in the the existential check gadget, will take the play from
	$(s_z,\cdot)$ to $(s^0,\cdot)$ and from $(s_{\overline{z}},\cdot)$ to
	$(s',\cdot)$. Now if Adam has cheated by taking the play to
	$(s_z,\cdot)$ when $c_i$ was not zero, then Eve---by playing
	$\texttt{bail}$---can go from $(s_z,c^+_i)$ to the initial observation
	with weight $-1$ and close a good abstract cycle. (Recall that a
	good abstract cycle is \textbf{good for Eve}: by repeatedly closing
	good cycles, Eve can win the MPG. Thus, closing a good cycle is our
	way of punishing Adam.) If Adam cheated by taking the play to
	$(s_{\overline{z}},\cdot)$ when $c_i$ was indeed zero, Eve can go (with
	the same action) from $(s_{\overline{z}},c^-_i)$ to the initial
	observation with weight $0$ again and close a good abstract cycle.
	Indeed, Adam can escape the zero-check gadget by choosing a non-proper
	successor. We will shortly explain why this is not a viable option for
	him.

\begin{figure}
\begin{center}
\begin{tikzpicture}[node distance=0.5cm]	
	\node[state](A){$c_i^+$};
	\node[state](B)[right=of A]{$c_i^-$};
	\node[state](C)[right=of B]{$\alpha^+$};
	\node[state](D)[right=of C]{$\alpha^-$};
	\node[fit=(A) (B) (C) (D),label=right:$s$](o){};

	\node[state,below= 1.2cm of A, xshift=-.3cm](A2){$c_i^+$};
	\node[right=of A2](B2){\dots};
	\node[fit=(A2) (B2),label=right:$s_z$](oz){};

	\node[state,below= 2.2cm of C, xshift=.3cm](B3){$c_i^-$};
	\node[right=of B3](C3){\dots};
	\node[fit=(B3) (C3),label=right:$s_{\overline{z}}$](onz){};
	\node[state](I)[left= 2cm of A,initial]{$q_I$};
	\node[fit=(I)]{};

	\node[below=1.8cm of oz, draw, densely dotted,
	rectangle,label=right:$s^0$,inner sep=.5cm](X){};
	\node[below=0.8cm of onz, draw, densely dotted,
	rectangle,label=right:$s'$,inner sep=.5cm](XX){};

	\path[]
	(o) edge node[el,swap,pos=0.9]{$\Sigma',0$} (oz)
	(o) edge node[el]{$\Sigma',0$} (onz)
	(A2) edge node[el,swap,pos=0.7]{$\texttt{bail},-1$} (I)
	(B3) edge[bend left,looseness=1.5] node[el]{$\texttt{bail},0$} (I)
	(oz) edge node[el,swap,pos=0.8]{$\Sigma'\setminus \{\texttt{bail}\},0$} (X)
	(onz) edge node[el,swap]{$\Sigma'\setminus \{\texttt{bail}\},0$} (XX)
;
\end{tikzpicture}
\end{center}
\caption{Universal observation gadget for $(s,0\mathrm{CHK},c_i,s',s^0)$ instruction.
	Transitions to $s',s^0$ observations are weighted as with the existential
	observation gadget.}
\label{fig:0chk-gadget-adam}
\end{figure}
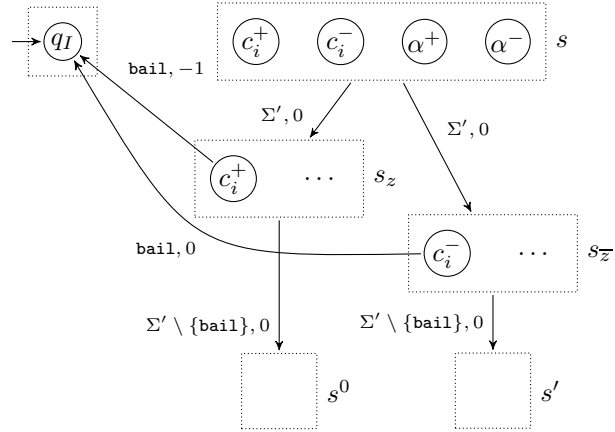

	\item \paragraph*{Stopping Simulator}
	It should be clear also from the gadgets, that the opponent of Simulator
	has no incentive to interrupt the simulation if there is no cheat. Doing
	so is actually beneficial to Simulator, who can get a function-action
	sequence which makes him win $\Gamma$.

	Finally, $\Delta$ also contains self-loops at all $(s_A, \cdot)$ with
	all $\Sigma'$ and with $0$ weights and at all $(s_R, \cdot)$ with all
	$\Sigma'$ and with $-1$ weights. Thus, if the play reaches the
	observation representing state $s_F$ or $s_R$ from $M$ then Simulator
	will be able to force function-action sequences which allow him to win
	$\Gamma$.

	\item \paragraph*{Making Adam play properly} We will now explain the
	idea behind observation gadget $\{a^+,a^-\}$.
	Note that Adam could break Eve's simulation of an accepting run by
	declaring the value of functions from $\Gamma$, which are actually
	our means of encoding the values of the counters, to be $+\infty$ (or at
	least some subset of the values of the functions). We describe how we
	obtain the final set of actions for the constructed game and mention the
	required transitions from every observation in the game so that it is
	not in the interest of Adam to do the latter.

	Denote by $(o,q_i)$ the $i$-th state in observation $o$. Observe that in
	our construction we need at most $10$ states per observation: two copies
	of every counter state and two additional step counters. The full action
	set in the game is defined as follows $\Sigma := \Sigma' \cup \{q_i \st
	0 \le i < 10\} \cup \{\texttt{ex}\}$. For every observation $o$ in $G$
	we add the transitions $((o,q_i),q_i,a^-),((o,q_j),q_i,a^+)$ for all
	$q_i,q_j$ in $o$ where $q_i \neq q_j$.  To finish, we add the self-loops
	$(a^+,\sigma,a^+)$ and $(a^-,\sigma,a^-)$, with weights $+1$ and $-1$
	respectively, as part of $\Delta$ for all $\sigma \in \Sigma$.  Clearly,
	Adam cannot choose anything other than proper $\sigma$-successors in
	$\Gamma$ or he gives Eve enough information for her to win the game.
	Indeed, if the game is currently at observation $o$ and Adam has chosen
	a non-proper successor, then Eve knows some state $q_i \in o$ does not
	currently hold the token.  Hence, she can play action $q_i$ and be sure
	to reach the state $a^+$ where she will win the game.

	\item \paragraph*{Bound on the length of the simulation}
	All that remains is to show how we allow the opponent of Simulator to stop
	the simulated run of $M$ in case Simulator exhausts the number of
	instructions he initially declared would be used to accept or reject.
	The $\texttt{ex}$ action is used in transitions
	$((o,\alpha^+),\texttt{ex},q_I) \in \Delta$ for all observation
	gadgets $o$ in the universal side of the construction. This allows Eve
	to stop Adam (who is playing Simulator) in case he tries to simulate
	more steps than he said were required for $M$ to reject. (Indeed, Adam
	may instead choose to move to another observation besides $\{q_I\}$ on
	$\texttt{ex}$, but then he reveals to Eve that some state $q_i$ in the
	following observation cannot hold the token, and she will then play $q_i
	\in \Sigma$ to win from there.) Similarly, in
	the existential part of the construction, we add a transition
	$((o,\alpha^-),\sigma,q_I)$ for all observation gadgets $o$ and all
	$\sigma \in \Sigma$, which lets Adam stop Eve's simulation if she tries
	to cheat in the same way.

	Finally, to have the game be limited observation we let all missing
	$\sigma$-transitions on the existential (resp. universal) side of the
	simulation go to a sink state in which Adam (Eve) wins.

	\item \paragraph*{Correctness}
	Now, let us prove the correctness of the overall construction. Assume
	that $M$ has an accepting or rejecting run. Then, Simulator, by
	simulating faithfully the run of $M$ has an observation-based strategy that allows him to
	force abstract paths which induce good or bad abstract cycles depending
	on who is simulating. Clearly, in this case even if the opponent decides
	to interrupt the simulation $M$ at a zero-check gadget, he will only
	be helping Simulator.

	If $M$ has no accepting or rejecting run, then by simulating the machine
	faithfully, Simulator will be generating cycles in the control state of
	the machine and such abstract paths are ``mixed'' because of concrete
	paths between corresponding $\alpha^-,\alpha^+$ states. Cheating does
	not help him either since after the opponent catches him cheating and
	restarts the simulation of the machine (by returning to the initial
	observation), the corresponding path is losing for him.
\end{proof}

It follows from Proposition~\ref{pro:superThm} that determining if a given
limited-observation MPG is forcibly terminating is as hard as determining if a
given $4$CM halts (by accepting or rejecting). As the latter problem is known to
be undecidable, we obtain the following result.
\begin{theorem}[Class membership]\label{thm:undec-ft}
	Let $G$ be an MPG with limited observation. Determining if $G$ is
	forcibly terminating is undecidable.
\end{theorem}

Although determining if Eve wins a forcibly terminating MPG is decidable,
Proposition~\ref{pro:superThm} implies the problem is extremely hard. Let
\R~denote the class of all decision problems solvable by a Turing machine. We
say a problem is \R-complete under polynomial reductions if it is decidable and
if any decidable problem reduces to it in polynomial time.
\begin{theorem}[Winner determination]\label{thm:r-complete}
	Let $G$ be a forcibly terminating MPG.
	Determining if Eve wins $G$ is $\R$-complete under polynomial reductions.
\end{theorem}
\begin{proof}
	For decidability,
	Lemma~\ref{lem:boundedPlays} implies that an alternating Turing Machine
	simulating a play on $\Gamma$ will terminate. Regarding hardness, we
	will argue that any decidable problem reduces to winner determination of
	forcibly terminating MPGs via our $4$CM simulation. Indeed, it is known
	that any given Turing machine can be simulated by a $4$CM of polynomial
	size with respect to the size of the original Turing machine. Also,
	given a decidable problem, we know there exists a Turing machine which,
	given any instance of the problem, always halts and outputs a positive
	or negative answer. We construct, from the Turing machine and a given
	instance of the problem, the corresponding $4$CM and, in turn, the
	corresponding limited-observation MPG as in the proof of
	Proposition~\ref{pro:superThm}. Since the original Turing machine always
	halts, the game is guaranteed to be forcibly terminating. Now, it should
	be clear that in the constructed game Eve wins if and only if the $4$CM
	accepts if and only if the Turing machine accepts the instance. As both
	the construction of the CM and the game are feasible in polynomial time,
	the result follows.
\end{proof}

\subsection{Modifications for Theorem~\ref{thm:liminf}}
\label{sec:proof-liminf}
To prove Theorem~\ref{thm:liminf} we reduce from the non-termination problem for
$2$CMs using a construction similar to the one used in the proof of
Proposition~\ref{pro:superThm}.  Given a $2$CM $M$, we construct a
game $G$ as in the proof of Proposition~\ref{pro:superThm}, with the following
adjustments:
\begin{itemize}[nolistsep]
	\item We only consider the universal side of the simulation, but allow
		both players to exit the initial pumping phase into it;
	\item The observation corresponding to the accept state of $M$ is a sink
		state winning for Adam;
	\item The $\alpha^-$ states are replaced with $\beta$ states which have
		transitions to other $\beta$ states of weight $0$ except in one
		case specified below;
	\item The pumping gadget has self loops of weights $0,0,-1$ and the
		transition from $b^+$ to $\beta$ has weight $-1$ if Eve exits
		and weight $0$ if Adam exits;
	\item The \texttt{ex} transition also goes from $\beta$ states to $q_I$.
\end{itemize}

Suppose the counter machine halts in $N$ steps. The observation-based winning
strategy for Adam is as
follows. Exit the pumping gadget after $N$ steps and faithfully simulate the
counter machine. Suppose Eve can beat this strategy. If she allows a faithful
simulation for $N$ steps then Adam reaches a sink state and wins, so Eve must
play \texttt{ex} within $N$ steps of the simulation. Let us consider each
cycle of at most $2N$ steps. If she waits for Adam to exit the pumping gadget
then the number of steps in the simulation is less than the number of steps in
the pumping gadget, so a negative cycle is closed. On the other hand if she
exits the pumping gadget before $N$ steps then the cycle through the $\beta$
states has negative weight. In both cases, a negative cycle is closed in at
most $2N$ steps, so the limit average is bounded above by $-\frac{1}{2N}$.
Thus this strategy is winning for Adam.

Now suppose the counter machine does not halt. The (infinite memory)
observation-based winning strategy for Eve is as follows. For increasing $n$,
exit the pumping gadget after $n$ steps and faithfully simulate (\ie call any,
and only, cheats of Adam) the counter machine for $n$ steps. Then play
\texttt{ex} and increase $n$.  Cheating in the simulation does not benefit Adam,
so we can assume Adam faithfully simulates the counter machine. Likewise, if Eve
always waits until the number of steps in the simulation exceeds the number of
steps in the pumping gadget, then there is no benefit for Adam to exit the
pumping gadget. However if the play proceeds as Eve intends then the weight of
the path through the $\alpha^+$ states is non-negative and although the weight
through the $\beta$ states is negative, the limit average is $0$. Thus the
strategy is winning for Eve.\qed

\section{Forcibly First Abstract Cycle Games}
\label{sec:adeqpure-inc}
In this section and the next we consider restrictions of forcibly terminating
games in order to find sub-classes with more efficient algorithmic bounds.
The negative algorithmic results from the previous section largely
arise from the fact that the abstract cycles required to determine the winner
are not necessarily simple cycles. Our first restriction of forcibly terminating
games is the restriction of the abstract cycle-forming game to simple cycles.

More precisely, let $G = ( Q,q_I,\Sigma,\Delta,w,\Obs )$
be an MPG with limited observation and $\Gamma = (
\Pi,\Sigma,f_I,\delta,\teve,\tadam )$ be the
associated reachability game. Define $\Pi' \subseteq \Pi$ as the set of all
sequences $f_0 \sigma_0 f_1 \sigma_1 \dots f_n \in \Pi$ such that
$\supp(f_i) \neq \supp(f_j)$ for all $0 \le i < j < n$ and denote by $\Gamma'$
the reachability game $( \Pi', \Sigma, f_I, \delta', \teve', \tadam'
)$ where $\delta'$ is $\delta$ restricted to $\Pi'$, $\teve' = \teve
\cap \Pi'$, and $\tadam' = \tadam \cap \Pi'$.

\begin{definition}
	An MPG with limited observation is \emph{forcibly first abstract
	cycle} (or forcibly FAC) if in the associated reachability game
	$\Gamma'$ either Adam has a winning strategy to reach states in
	$\tadam'$ or Eve has a winning strategy to reach states in $\teve'$.
\end{definition}

One immediate consequence of the restriction to simple abstract cycles is that
the bound in Lemma~\ref{lem:boundedPlays} is at most $|\Obs|$.  In particular an
alternating Turing Machine can, in linear time, simulate a play of the
reachability game and decide which player, if any, has a winning strategy. Hence
the problems of deciding if a given MPG with partial observation is forcibly FAC
and deciding the winner of a forcibly FAC game are both solvable in $\PSPACE$.
The next results show that there is a matching lower bound for both these
problems.

\begin{theorem}[Class membership]
\label{thm:isadeqpure}
	Let $G$ be an MPG with limited observation. Determining if $G$ is
	forcibly FAC is \PSPACE-complete.
\end{theorem}

\begin{proof}
	For \PSPACE~membership we observe that a linear bounded alternating
	Turing Machine can decide whether one of the players can force to reach
	$\teve'$ or $\tadam'$ in $\Gamma'$. To show hardness we use a reduction
	from the True Quantified Boolean Formula (TQBF) problem. Given a
	\emph{fully quantified} Boolean formula $\Psi = \exists x_0 \forall x_1
	\dots \mathcal{Q} x_{n-1} (\Phi)$, where $\mathcal{Q} \in
	\{\exists,\forall\}$ and $\Phi$ is a Boolean formula expressed in
	\emph{conjunctive normal form} (CNF), the TQBF problem asks whether
	$\Psi$ is true or false. The TQBF problem is known to be
	\PSPACE-complete~\cite{sm73}.

	This problem is often rephrased as a game between Adam and Eve. In this
	game the two players alternate choosing values for each $x_i$ from
	$\Phi$. Eve wins if the resulting evaluation of $\Phi$ is true while
	Adam wins if it is false. We simulate such a game with the use of
	``diamond'' gadgets that allow Eve to choose a value for existentially
	quantified variables by letting her choose the next observation.
	Similarly, the same gadget---except for the labels on the transitions,
	which are completely non-deterministic in the following case---allow
	Adam to choose values for variables that are universally quantified.

	We construct a game $G_\Psi = ( Q, q_I, \Sigma, \Delta, w, \Obs
	)$ in which there are no concrete negative cycles, hence it
	follows from Lemma~\ref{lem:caVa} that there are no bad cycles. The game
	will thus be forcibly FAC if and only if Eve is able to force good
	cycles. If Eve is unable to prove the QBF is true, Adam will be able to
	avoid such plays.  For this purpose, the ``diamond'' gadgets employed
	have two states per observation. This will allow two disjoint concrete
	paths to go from the initial state $q_I$ through the whole arena and
	form a simple abstract cycle that is either good or not good depending
	on where the cycle started from.

	Concretely, let $x_1$ be a universally quantified variable from $\Psi$.
	We add a gadget to $G_\Psi$ consisting of eight states grouped into four
	observations: $\{b_0^-,b_0^0\}$, $\{\overline{x_1},\overline{z_1}\}$,
	$\{x_1,z_1\}$, $\{b_1^-,b_1^0\}$.  We also add the following
	transitions:
	\begin{itemize}[nolistsep]
		\item from $b_0^-$ to $\overline{x_1}$ and $x_1$, $b_0^0$ to
			$\overline{z_1}$ and $z_1$, with all $\Sigma$ and weight
			$0$;
		\item from $\overline{x_1}$ and $x_1$ to $b_1^-$,
			$\overline{z_1}$ and $z_1$ to $b_1^0$, with all $\Sigma$
			and the first two with weight $-1$ while the last two
			have weight $0$.
	\end{itemize}
	Figure~\ref{fig:formula-game} shows the universal ``diamond'' gadget
	just described. The observation $\{\overline{x_1},\overline{z_1}\}$
	corresponds to the variable being given a false valuation, whereas the
	$\{x_1,z_1\}$ observation models a true valuation having been picked.
	Observe that the choice of the next observation from $\{b_0^-,b_0^0\}$
	is completely non-deterministic, \ie Adam chooses the valuation for
	this variable.

	For existentially quantified variables, the first set of transitions
	from the gadget is slightly different. Let $x_i$ be an existentially
	quantified variable in $\Psi$, then the upper part of the gadget
	includes transitions from $b_i^-$ to $\overline{x_i}$ and from $b_i^0$
	to $\overline{z_i}$ with action symbol $\lnot x_i$ and weight $0$; as
	well as transitions from $b_i^-$ to $x_i$ and from $b_i^0$ to $z_i$ with
	action symbol $x_i$ and weight $0$.

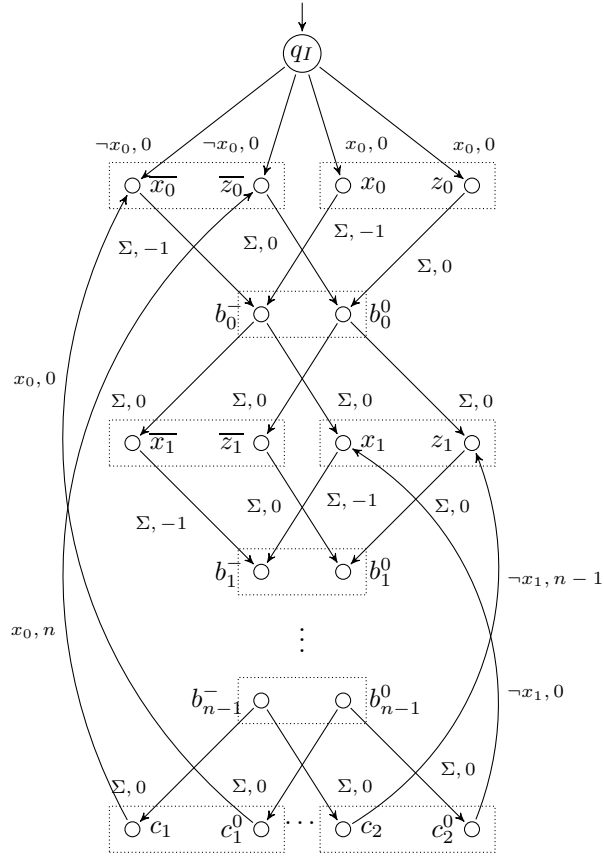
\begin{figure}
\begin{center}
\begin{tikzpicture}[bend angle=40, node
			distance=1.5cm,v/.style={draw,circle,minimum size=2pt,inner sep=2pt}]
	\node[state,initial above](A){$q_I$};
	\node[v,below left = 1.5cm and 2cm of A,label=right:$\overline{x_0}$](B){};
	\node[v,right = of B,label=left:$\overline{z_0}$](B0){};
	\node[v,below right = 1.5cm and 2cm of A,label=left:$z_0$](C0){};
	\node[v,left = of C0,label=right:$x_0$](C){};
	\node[v,below = of B0,label=left:$b_0^-$](D){};
	\node[v,below = of C,label=right:~$b_0^0$](D0){};

	\node[fit=(B) (B0)]{};
	\node[fit=(C) (C0)]{};
	\node[fit=(D) (D0)]{};

	\path
	(A) edge node[el,swap,pos=0.85]{$\lnot x_0,0$} (B)
	(A) edge node[el,swap,pos=0.85]{$\lnot x_0,0$} (B0)
	(A) edge node[el,pos=0.85]{$x_0,0$} (C)
	(A) edge node[el,pos=0.85]{$x_0,0$} (C0)
	(B) edge node[el,swap,pos=0.35]{$\Sigma,-1$} (D)
	(B0) edge node[el,swap,pos=0.3]{$\Sigma,0$} (D0)
	(C) edge node[el,pos=0.2]{$\Sigma,-1$} (D)
	(C0) edge node[el]{$\Sigma,0$} (D0);

	\node[v,below = of D,label=left:$\overline{z_1}$](F0){};
	\node[v,left = of F0,label=right:$\overline{x_1}$](F){};
	\node[v,below = of D0,label=right:$x_1$](G){};
	\node[v,right = of G,label=left:$z_1$](G0){};
	\node[v,below = of F0,label=left:$b_1^-$](H){};
	\node[v,below = of G,label=right:~$b_1^0$](H0){};

	\node[fit=(F) (F0)]{};
	\node[fit=(G) (G0)]{};
	\node[fit=(H) (H0)](ho){};

	\path
	(D) edge node[el,swap,pos=0.85]{$\Sigma,0$} (F)
	(D0) edge node[el,swap,pos=0.85]{$\Sigma,0$} (F0)
	(D) edge node[el,pos=0.85]{$\Sigma,0$} (G)
	(D0) edge node[el,pos=0.85]{$\Sigma,0$} (G0)
	(F) edge node[el,swap]{$\Sigma,-1$} (H)
	(F0) edge node[el,swap,pos=0.35]{$\Sigma,0$} (H0)
	(G) edge node[el,pos=0.3]{$\Sigma,-1$} (H)
	(G0) edge node[el,pos=0.35]{$\Sigma,0$} (H0);

	\node[below = 0.1cm of ho](dots){\vdots};
	\node[v,below = of H,label=left:$b_{n-1}^-$](CL){};
	\node[v,below = of H0,label=right:~$b_{n-1}^0$](CL0){};
	\node[v,below = of CL,label=left:$c_1^0$](c10){};
	\node[v,left = of c10,label=right:$c_1$](c1){};
	\node[v,below = of CL0,label=right:$c_2$](c2){};
	\node[v,right = of c2,label=left:$c_2^0$](c20){};
	\node[below = 2cm of dots]{\ldots};

	\node[fit=(CL) (CL0)]{};
	\node[fit=(c1) (c10)]{};
	\node[fit=(c2) (c20)]{};

	\path
	(CL0) edge node[el,swap,pos=0.85]{$\Sigma,0$} (c10)
	(CL) edge node[el,swap,pos=0.85]{$\Sigma,0$} (c1)
	(CL0) edge node[el,pos=0.7]{$\Sigma,0$} (c20)
	(CL) edge node[el,pos=0.85]{$\Sigma,0$} (c2)

	(c1) edge[bend left] node[el,pos=0.3]{$x_0,n$} (B0)
	(c10) edge[bend left] node[el,pos=0.7]{$x_0,0$} (B)
	(c2) edge[bend right] node[el,swap,pos=0.7]{$\lnot x_1, n-1$} (G0)
	(c20) edge[bend right] node[el,swap,pos=0.3]{$\lnot x_1, 0$} (G);
\end{tikzpicture}
\end{center}
\caption{Corresponding game for QBF $\exists x_0 \forall x_1 \dots ( \lnot x_0)
	 \wedge (x_1) \dots$}
\label{fig:formula-game}
\end{figure}

	A play in $G_\Psi$ traverses gadgets for all the variables from the QBF
	and eventually gets to the observation $\{b_{n-1}^-,b_{n-1}^0\}$ where
	the assignment of values for every variable has been simulated. At this
	point we want to check whether the valuation of the variables makes
	$\Phi$ true. We do so by allowing Adam to choose the next observation
	(corresponding to one of the clauses from the CNF formula $\Phi$) and
	letting Eve choose a variable from the clause (which might be negated).
	Let $x_i$ (resp. $\overline{x_i}$) be the variable chosen by Eve, in
	$G_\Psi$ the next observation will correspond to closing a good abstract
	cycle if and only if the chosen valuation of the variables for $\Psi$
	assigns to $x_i$ a true (false) value. For this part of the construction
	we have $2m$ states grouped in $m$ observations, where $m$ is the
	number of clauses in the formula. The lower part of
	figure~\ref{fig:formula-game} shows the clause observations we just
	described.

	Denote by $\{ c_i, c_i^0 \}$ the observation associated to clause $c_i$.
	The game has transitions from $c_i$ to $x_i$ (or $\overline{x_i}$) and
	from $c_i^0$ to $z_i$ ($\overline{z_i}$) with action symbol $x_i$
	($\lnot x_i$) and weight $n-i$ for the first, $0$ for the latter, if and
	only if the clause $c_i$ includes the (negated) variable
	$x_i$.\footnote{All missing transitions for $G_\Psi$ to be complete go
		to a dummy state with a negative and $0$-valued
	non-deterministic transitions.}

	After Eve and Adam have chosen values for all variables (and the game
	reaches observation $\{b_{n-1}^-,b_{n-1}^0\}$) there are two concrete
	paths corresponding with the current play: one with payoff $0$ and one
	with payoff $-n$.  When Adam has chosen a clause and Eve chooses a
	variable $x_i$ from the clause, the next observation is reached with
	both concrete paths having payoffs $0$ and $-i$. Observe, however, that
	if we consider the suffix of said concrete paths starting from
	$\{x_i,z_i\}$ or $\{\overline{x_i},\overline{z_i}\}$---depending on
	which valuation the players chose---both payoffs are $0$. Indeed, if
	the observation was previously visited, \ie Eve has proven the clause
	to be true, then a good cycle is closed. On the other hand, if the
	observation has not been visited previously, then Eve has no choice but
	to keep playing and the play thus reaches observation $\{b_i^-,b_i^0\}$.
	We note that traversing the lower part of our ``diamond'' gadgets
	results in a \emph{mixed} payoff of $-1$ and $0$ and since
	$\{b_i^-,b_i^0\}$ has already been visited, a cycle is closed that is
	not good. To summarize, either a good cycle is closed when moving from
	$\{b_{n-1}^-,b_{n-1}^0\}$ to $\{x_i,z_i\}$ (or, respectively,
	$\{\overline{x_i},\overline{z_i}\}$) if the latter observation had been
	visited before; or a bad cycle is closed on the next step when moving to
	$\{b_i^-,b_i^0\}$.

	Therefore, if $\Psi$ is true then Eve has a strategy to make the first
	cycle closed be a good one, so $G_{\Psi}$ is forcibly FAC. Conversely,
	if $\Psi$ is false then Adam has a strategy to make the first cycle
	formed be not good (mixed, in fact).  Hence $G_\Psi$ is not forcibly
	FAC.
\end{proof}

We can slightly modify the above construction in such a way that if the game
does not finish when the play returns to a variable then Adam can close a bad
cycle (instead of just being able to force a mixed cycle).  This results in a
forcibly FAC game that Eve wins if and only if the formula is satisfied. Hence,

\begin{theorem}[Winner determination]\label{thm:apwd}
	Let $G$ be a forcibly FAC MPG.
	Determining if Eve wins $G$ is \PSPACE-complete.
\end{theorem}
\begin{proof}
	We describe the modifications required to the construction used in the
	proof of Theorem~\ref{thm:isadeqpure}.
	
	First, we augment every observation with $2n$ states corresponding to
	variables from $\Phi$ and their negation (say, $y_i$ and $\overline{y_i}$
	for $0 \le i < n$).
	
	We then add transitions from every new state $y_i$ ($\overline{y_i}$) to
	its counterpart in the next observation so as to form $2n$ new
	disjoint cycles going from $q_I$ through the whole construction. (Note
	that, up to this point $\Plays(G_\Psi)$ remains unchanged. That is, the
	set of abstract paths in the game constructed for the proof of
	Theorem~\ref{thm:isadeqpure} is the same as the set of abstract paths in
	the present game.) These new transitions all have weight zero except for
	a few exceptions:
	\begin{itemize}[nolistsep]
		\item the transition corresponding to the lower part of the
			gadget which represents the variable itself, \ie the
			transition from augmented observation
			$\{x_i,z_i,\dots\}$ to $\{b_i^-,b_i^0,\dots\}$ (resp.
			$\{\overline{x_i},\overline{z_i},\dots\}$ to
			$\{b_i^-,b_i^0,\dots\}$) has weight of $+1$ for the
			$y_i$-transition ($\overline{y_i}$-transition);
		\item outgoing transitions from clause observations have weight
			$-1$ on the $y_i$-transition going to the
			$x_i$-gadget---\ie if we let $y_j$ be one of the new
			states in the clause observation and $y'_j$ the
			corresponding state in the $x_i$-gadget, then
			$w(y_j,\sigma,y'_j) = -1$; and
		\item at every $\{x_i,z_i,\dots\}$ and
			$\{\overline{x_i},\overline{z_i},\dots\}$ augmented
			observation, Adam is allowed to resolve non-determinism
			by going back to $q_I$---\ie in these observations we
			add a transition from $y_i$ and $\overline{y_i}$,
			respectively, back to the initial state.
	\end{itemize}
	
	Let us argue that the game is forcibly FAC and that the QBF instance is
	true if and only if Eve wins the reachability game associated with the
	constructed MPG.
	When the play reaches $\{x_i, z_i, \dots\}$ (or
	$\{\overline{x_i}, \overline{z_i}, \dots\}$) after Eve and Adam choose
	values for all the variables and after she has chosen a variable from a
	clause given by Adam, then the concrete path ending at $y_i$ (resp.
	$\overline{y}_i$) has weight $0$ if the observation was previously
	visited, and weight $-1$ if it was not. The concrete paths ending at all
	the other new states have weight $0$ or $+1$ depending on the choices
	made by the players. Concrete paths ending at $x_i$ and $z_i$ states are
	as before (mixed if the observation has not been witnessed, and good
	otherwise).  Thus if the observation was previously visited, then the
	cycle closed is good as before. If the observation was not previously
	visited, then Adam can now choose to play to $q_I$ from $y_i$
	($\overline{y_i}$) and close a bad cycle (of weight $-1$). Note that if
	Adam chooses to play to $q_I$ before the clause gadgets are reached then
	he will only be closing good cycles. Following the same argument as
	before, if $\Psi$ is true then Eve has a winning strategy and if $\Psi$
	is false then Adam has a winning strategy. So $G_{\Psi}$ is forcibly FAC
	and Eve wins if and only if $\Psi$ is true.
\end{proof}

It also follows from the $|\Obs|$ upper bound on plays in $\Gamma'$ that there
is an exponential upper bound on the memory required for a winning strategy for
either player.
Furthermore, we can show this bound is tight---the games constructed in the
proof of Theorem~\ref{thm:apwd} can be used to show that there are forcibly FAC
games that require exponential memory for winning strategies.

\begin{theorem}[Exponential memory determinacy]
	\label{thm:expMemory}
	One player always has a winning observation-based strategy with
	exponential memory in a forcibly FAC MPG.  Further, for any $n \in
	\mathbb{N}$ there exists a forcibly FAC MPG, of size polynomial in $n$,
	such that any winning observation-based strategy has memory at least $2^n$.
\end{theorem}
\begin{proof}
	For the upper bound we observe that plays in $\Gamma'$ are bounded in
	length by $|\Obs|$. It follows that the strategy constructed in
	Theorem~\ref{thm:stratTransfer} has memory at most $|\Sigma|^{|\Obs|}$.

	For the lower bound, consider the forcibly FAC game $G_n$ constructed in
	the proof of Theorem~\ref{thm:apwd} for the formula
	\[
		\varphi_n = \forall x_1 \forall x_2 \dots \forall x_n \exists
		y_1 \dots \exists y_n.\bigwedge_{i=1}^n (x_i \vee \neg y_i)
		\wedge
		(\neg x_i \vee y_i).
	\]
	As $\varphi_n$ is satisfied, Eve wins $G_n$. Now consider any
	observation-based strategy
	for Eve with memory $<2^n$. As there are $2^n$ possible assignments for
	the values of $x_1, \dots x_n$ it follows there are at least two
	different assignments of values such that Eve makes the same choices in
	the game. Suppose these two assignments differ at $x_i$ and assume w.l.o.g.
	that Eve's choice is at $(n+i)$-th gadget to play to $y_i$.  Then Adam
	can win the game by choosing values for the universal variables that
	correspond to the assignment which sets $x_i$ to \textsf{false}, and
	then playing to the clause $(x_i, \vee \neg y_i)$.  Thus any winning
	observation-based strategy for Eve must have size at least $2^n$.

	In a similar way the game defined by the formula $\neg \phi_n$ is won by
	Adam, but any winning observation-based strategy must have size at least $2^n$.
\end{proof}

\section{First Abstract Cycle Games}\label{sec:pure-inc}
We now consider a structural restriction that guarantees $\Gamma'$ is
determined. Recall that to any limited-observation MPG $G = (
Q,q_I,\Sigma,\Delta,w,\Obs )$ we associate a reachability game $\Gamma =
( \Pi,\Sigma,f_I,\delta,\teve,\tadam )$ and that $\Gamma'$ is the
restriction of $\Gamma$ to simple function-action sequences (with respect to the
supports). That is, $\Pi'$ is the set of all sequences $f_0 \sigma_0 f_1
\sigma_1 \dots f_n \in \Pi$ such that $\supp(f_i) \neq \supp(f_j)$ for all $0
\le i < j < n$ and the other components of $\Gamma'$ are the corresponding
restrictions of $\Gamma$ to $\Pi'$.

\begin{definition}
	An MPG with limited observation is a \emph{first abstract cycle game}
	(FAC) if in the associated reachability game $\Gamma'$ all leaves are
	in $\tadam' \cup \teve'$.
\end{definition}

Intuitively, in an FAC game $G$ all simple abstract cycles (that can be formed) are
either good or bad. Since $\Gamma'$ is a full-observation finite reachability
game, $G$ is determined. Thus, by Theorem~\ref{thm:stratTransfer}, we get that in
every FAC one of the two players has a winning finite-memory observation-based
strategy. However, we can show an even stronger result holds:
one of them has a winning positional
observation-based strategy.

\begin{theorem}[Positional determinacy]
\label{thm:posDet}
	One player always has a positional winning observation-based strategy in
	an FAC MPG.
\end{theorem}
\begin{proof}
	It follows then from Corollary~\ref{cor:quelleSauce} that any cyclic
	permutation of a good cycle is also good and any cyclic permutation of a
	bad cycle is also bad. Together with Lemma~\ref{lem:caVa}, this implies
	the abstract cycle-forming games associated with FAC games can be seen
	to satisfy the following three assumptions:
	\begin{inparaenum}[(1)]
		\item A play stops as soon as an abstract cycle is formed;
		\item The winning condition and its complement are preserved under cyclic
			permutations; and
		\item The winning condition and its complement are preserved under
			interleavings.
	\end{inparaenum}
	These assumptions were shown in~\cite{ar14} to be sufficient for winning
	positional strategies to exist in any game.\footnote{These conditions
		supersede those of~\cite{bsv04} which were shown in~\cite{ar14}
		to be insufficient for positional strategies.}
\end{proof}

As we can check in polynomial time if a positional observation-based strategy is
winning in an FAC MPG, we immediately have:
\begin{theorem}[Winner determination]\label{thm:windet-fac}
	Let $G$ be an FAC MPG. Determining if Eve wins
 	$G$ is in $\NP \cap \coNP$.
\end{theorem}

A path in $\Gamma'$ to a leaf not in $\tadam' \cup \teve'$ provides a short
certificate to show that an MPG with limited observation is not FAC. Thus
deciding if an MPG is FAC is in $\coNP$. A matching lower bound can be obtained
using a reduction from the complement of the Hamiltonian cycle problem.

\begin{theorem}[Class membership]
\label{thm:pureCM}
	Let $G$ be an MPG with limited observation. Determining if $G$ is FAC
	is \coNP-complete.
\end{theorem}
\begin{proof}
	For \coNP~membership, one can guess a large enough simple abstract cycle
	$\psi$ and (in polynomial time with respect to $Q$) check that it is
	neither good nor bad. To show \coNP-hardness we use a reduction from the
	complement of the Hamiltonian Cycle problem.

	Given graph $\mathcal{G} = ( V, E )$ where $V$ is the set of
	vertices and $E \subseteq V \times V$ the set of edges. We construct a
	directed weighted graph with limited observation $G = ( Q, q_I,
	\Sigma, \Delta, w, \Obs )$ where:
	\begin{itemize}[nolistsep]
		\item $Q = V \cup \{q_I, q_+, q_-\}$;
		\item $\Obs = \{ \{v\} \st v \in V \} \cup \{ \{q_-,q_+\},\{ q_I
			\} \}$;
		\item $\Sigma = V \cup \{\tau\}$;
		\item $\Delta$ contains transitions $(u,v,v)$ such that $(u,v)
			\in E$ and self-loops $(u,v',u)$ for all $(u,v') \notin
			E$, transitions (with all $\sigma$) from $q_I$ to both
			$q_+$ and $q_-$ and from these last two to all states $v
			\in V$, as well as $\tau$-transitions from every state
			$v \in V$ to $q_+$ and $q_-$;
		\item $w$ is such that all outgoing transitions from $q_+$ and
			$q_-$ have weight $1-|V|$, $(u,v,v)$ transitions where
			$(u,v) \in E$ have weight $+1$, $\tau$-transitions to
			$q_-$ from states $v \in V$ have weight $-1$ and all
			other transitions have weight $0$.
	\end{itemize}
	Notice that the only non-deterministic transitions in $G$ are those
	incident on and outgoing from the states $q_+$, $q_-$. Clearly, the only
	way for a simple abstract cycle to be not good and not bad (thus making
	$G$ not FAC) is if there is a path from $\{q_-,q_+\} \in \Obs$ that
	traverses $|V|$ unique observations and ends with a $\tau$-transition
	back at $\{q_-,q_+\}$.  Such a path corresponds to a Hamiltonian cycle
	in $\mathcal{G}$. If there is no Hamiltonian cycle in $\mathcal{G}$ then
	for any play $\pi$ in $G$, a bad cycle will be formed (hence, $G$ is
	FAC).
\end{proof}

\section{MPGs with Partial Observation}
\label{sec:pure-imp}
In the introduction it was mentioned that an MPG with partial observation can
be transformed into an MPG with limited observation. This translation allows
us to extend the notions of FAC and forcibly FAC games to the larger class
of MPGs with partial observation. In this section we will investigate the
resulting algorithmic effect of this translation on the decision problems
we have been considering.

The idea behind the translation is to take subsets of the observations and
restrict transitions to those that satisfy the limited-observation requirements.
More formally, given an MPG with partial observation $G = (
Q,\Sigma,\Delta,q_I,w,\Obs)$ we construct an MPG with limited observation
$G' = ( Q',\Sigma, \Delta', q_I', w',\Obs')$ where:
\begin{itemize}[nolistsep]
	\item $Q' = \{(q,K) \in Q \times 2^Q \st q \in K \text{ and } K
	\subseteq o \in \Obs\}$,
	\item $q_I' = (q_I,\{q_I\})$,
	\item $\Obs' = \big\{\{ (q,K) \st q \in K\} \st K \subseteq o\text{ for
		some }o \in \Obs\big\}$,
	\item $\Delta'$ contains the transitions $((q,K),\sigma,(q',K'))$ such that
		$(q,\sigma,q') \in \Delta$ and $K' = post_\sigma(K) \cap o$
		for some $o \in \Obs$, and
	\item $w'((q,K), \sigma, (q',K')) = w(q,\sigma,q')$ for all
		$((q,K),\sigma,(q',K')) \in \Delta'$.
\end{itemize}
It is folklore to show that this \emph{knowledge-based} subset construction
(also known as a belief construction) preserves winning strategies for Eve.  The
terms belief and knowledge are used to denote a state from any variation of the
classic ``Reif construction''~\cite{reif84} to turn a game with
partial observation into a game with full observation.  Other names for
similar constructions include ``knowledge-based subset construction'' (see
\eg~\cite{ddgrt10}). In this case the resulting game is not one with full
observation but one with limited observation.
\begin{theorem}[Equivalence]
	Let $G$ be an MPG with partial observation and $G'$ be the corresponding
	MPG with limited observation as constructed above. Eve has a winning
	observation-based strategy in $G$ if and only if she has a winning
	observation-based strategy in $G'$.
\end{theorem}

The result above is shown by proving that winning observation-based strategies
for Eve transfer between $G$ and $G'$. It is worth noting that an
observation-based strategy for Eve in $G$ can directly be used by her in $G'$.
Conversely, for her to use a strategy from $G'$ in $G$ she must keep in memory
the knowledge-based subset construction herself. Hence,
\begin{theorem}[Memory requirements]\label{thm:exp-mem-po}
	Let $G$ be a partial-observation MPG and $G'$ be the corresponding
	limited-observation MPG.  If a player has a finite-memory
	observation-based winning strategy in $G'$, then (s)he has a
	finite-memory observation-based winning strategy in $G'$ which requires
	exponentially more memory (on the size of $G$).
\end{theorem}

We say an MPG with partial observation is \emph{(forcibly) first belief cycle},
or FBC, if the corresponding MPG with limited observation is (forcibly) FAC.

\section{FBC and Forcibly FBC MPGs}
Our first observation is that FBC MPGs generalize the class of \emph{visible
weight games} studied in~\cite{ddgrt10}. An MPG with partial observation is
considered a visible weights game if its weight function satisfies the condition
that all $\sigma$-transitions between any pair of observations have the same
weight.  We base some of our results for FBC and forcibly FBC games on lower
bounds established for problems on visible weights games.

\begin{lemma}
	\label{lem:pure-visible}
	Let $G$ be a visible weights MPG with partial observation. Then $G$ is
	FBC.
\end{lemma}

We now turn to the decision problems we have been investigating throughout the
paper. Given the exponential blow-up in the construction of the game of
limited observation, it is not surprising that there is a corresponding
exponential increase in the complexity of the class membership problem.
\begin{theorem}[Class membership]
\label{thm:isimppure}
	Let $G$ be an MPG with partial observation.  Determining if $G$ is FBC
	is \coNEXP-complete and determining if $G$ is forcibly FBC is in
	\EXPSPACE\ and \NEXP-hard.
\end{theorem}

Membership of the relevant classes is straightforward, they follow directly
from the upper bounds for MPGs with limited observation and the (at worst)
exponential blow-up in the translation from games of partial observation to
games of limited observation. For hardness, we prove first the result for FBC
games and comment on the changes necessary for the construction to yield the
result for forcibly FBC games.

\begin{lemma}\label{lem:hardness-fbc}
	Let $G$ be an MPG with partial observation. Determining if $G$ is FBC is
	\coNEXP-hard.
\end{lemma}
\begin{proof}
	We reduce from the complement of the succinct Hamilton cycle
	problem: Given a Boolean circuit $C$ with $2N$ inputs, does the graph on
	$2^N$ nodes with edge relation encoded by $C$ have a Hamiltonian cycle?
	This problem is known to be \NEXP-complete~\cite{py86}.

	The idea is to simulate a traversal of the succinct graph in our MPG: if
	we make $2^N$ valid steps without revisiting a vertex of the succinct
	graph then that guarantees a Hamiltonian cycle. To do this, we start
	with a transition of weight $-2^N$ and add $1$ to all paths every time
	we make a valid transition.  We include a pair of transitions back to
	the initial state with weights $0$ and $-1$ and ensure this is the only
	transition that can be taken that results in paths of different weight.
	The resulting game then has a mixed lasso if and only if we can make
	$2^N$ valid transitions. If we encode the succinct graph vertex in the
	knowledge set then the definition of an FAC game will give us an
	automatic check if we revisit a vertex. In fact, we store several pieces
	of information in the knowledge sets of the observations: the current
	(succinct) graph vertex, the potential successor, and the evaluation of
	the edge-transition circuit up to a point.

	\begin{figure}
	\begin{center}
	\begin{tikzpicture}[node distance=0.5cm]

	\node [state] (v1) {$v_1$};
	\node [state,above right= of v1] (v2) {$v_2$};
	\node [state,below right= of v2] (v3) {$v_3$};

	\node [state,below=1cm of v1] (v1') {$\overline{v}_1$};
	\node [state,below right= of v1'] (v2') {$\overline{v}_2$};
	\node [state,above right= of v2'] (v3') {$\overline{v}_3$};

	\path
	(v1) edge node[el,swap]{$\tau^-$} (v3)
	(v1) edge node[el]{$\tau^+$} (v2)
	(v2) edge node[el]{$\chi,+1$} (v3)
	(v1') edge node[el]{$\tau^+$} (v3')
	(v1') edge node[el,swap]{$\tau^-$} (v2')
	(v2') edge node[el,swap]{$\chi,+1$} (v3');
		
	\node [state,below=1cm of v1'] (v0) {$v_0$};
	\node [state,below= of v0] (v0'){$\overline{v}_0$};
	\node [state,below=1cm of v3'] (x){$f$};
	\node [state,below= of x] (x'){$\overline{f}$};	

	\path
	(v0) edge node[el]{$\tau^+,\tau^-$} (x)
	(v0') edge node[el]{$\tau^+$} (x)
	(v0') edge node[swap,el]{$\tau^-$} (x')
    ;

	\node[fit = (v3) (v3') (v1) (v1') (v2) (v2') (v0) (v0'), dashed, label=above:$x \lor
	\lnot y$](gj){};

	\node [state,left=1.2cm of v1] (x2){$x$};
	\node [state,left=1.2cm of v1'] (x2'){$\overline{x}$};
	\node [state,left=1.2cm of v0] (y2){$y$};
	\node [state,left=1.2cm of v0'] (y2'){$\overline{y}$};

	\node[fit = (x2) (x2') (y2) (y2'), dashed, label=above:{$x,y$}]{};

	\path
	(x2) edge node[el]{$\sigma,-1$} (v1)
	(x2') edge node[el]{$\sigma,-1$} (v1')
	(y2) edge node[el,swap,pos=0.85]{$\sigma,-1$} (v0')
	(y2') edge node[el,pos=0.85]{$\sigma,-1$} (v0)
	;

	\node[state,right=1.2cm of x] (z){$z$};
	\node[state,right=1.2cm of x'] (z'){$\overline{z}$};

	\path
	(x) edge node[el]{$\sigma$} (z)
	(x') edge node[el]{$\sigma$} (z');

	\node[fit = (z) (z'), dashed, label=above:$z$]{ };

	\end{tikzpicture}
	\end{center}
    \caption{This is the partial-observation gadget to simulate $x \lor \lnot y
        = z$. Inside the gate gadget we also have on all states:
        self-loops with weight $+1$ on $\chi$, and $0$ on $\tau^{\pm}$;
        $0$-weight transitions to a sink on external actions.
        (Zero-weights have been omitted for clarity.)}
	\label{fig:example-or}
	\end{figure}
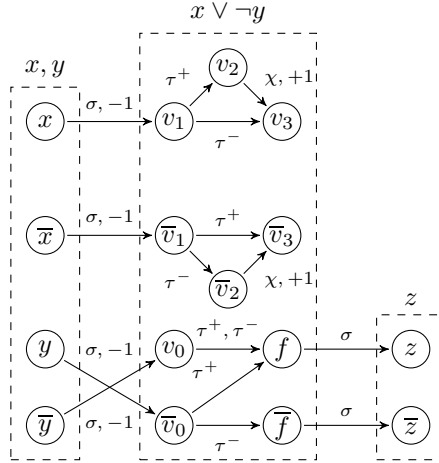

	\item \paragraph*{Simulating gates}
	The crucial technical trick used in our reduction is the construction of
	an observation gadget which simulates a logical gate.
	Figure~\ref{fig:example-or} depicts the gadget corresponding to $z = x
	\lor \lnot y$. We assume we have a knowledge set $K$ that is a subset of
	the states from the leftmost observation. Further, we assume $K$ induces
    a valid valuation of $x$ and $y$, \ie~$x \in K$ if and only if $\overline{x}
    \not\in K$
    and similarly
	for $y$. Denote by $K \models x \lor \lnot y$ the fact that the
	valuation of $x$ and $y$ makes the formula true. We also assume all
	concrete paths arriving at states in $K$ have the same weight. Now, by
	playing $\sigma$, Eve reaches the $x \lor \lnot y$ observation where she
	can play internal actions $\tau^-,\tau^+,\chi$. We claim
	observation $x \lor \lnot y$ allows concrete plays to reach the $z$
	observation with weight $0$ without creating a non-mixed belief lasso
    if and only if $K \models x \lor \lnot y$.
	%we consider all possible $K$ sets.
    The main idea is that Eve declares
	the truth value of $x$ using $\tau^+$ if it is true and $\tau^-$
	otherwise, she then plays $\chi$ to cancel the $-1$ weight seen upon
	entering the observation. For example, if $K = \{x,y\}$, Eve plays
	$\sigma$ and enters the gate observation with knowledge set
	$\{v_1,\overline{v}_0\}$. Then, Eve plays $\tau^+$ and
	one concrete path moves from $v_1$ to $v_2$, the other from
	$\overline{v}_0$ to $f$; Eve then plays $\chi$ and concrete paths reach
	$v_3$ and $f$, both with weight $0$; finally, she plays $\sigma$, and a
	concrete play reaches a sink or a concrete play reaches $z$ (as
	expected since $K \models x \lor \lnot y$). Crucially, the sequence of
    internal transitions on
    $\tau^\pm\chi$ induces a sequence of three distinct knowledge sets if and
    only if she declared the correct value of $x$. Otherwise, a lasso is formed.
    %. It is worth noting here that
	%our gadget actually enforces the following: Eve has a strategy to ensure
	%that she wins the partial-observation MPG or a concrete path reaches the
	%$z$ observation with weight $0$ if and only if $K \models x \lor \lnot
	%y$. Further, if she plays any other strategy, she ends up in a bad sink
	%for her and loses the game.
	
	We now describe the
	construction in detail.

	\begin{figure}
	\begin{center}
	\begin{tikzpicture}[semithick, node distance=1cm, v/.style={circle,minimum
		size=1pt, draw}, o/.style={rectangle,minimum size=1pt,
		draw,dashed}]

	\node[v,initial](s) {$S$};
	\node[o,below =of s](a) {$O_1$};
	\node[o,below =of a](b) {$G_0$ ($+N$)};
	\node[o,below right=of b](c) {$G_1$ ($+1$)};
	\node[right =of c](d) {$\cdots$};
	\node[o,right =of d](e) {$G_k$ ($+1$)};
	\node[o,above =of e](ch) {$Chk$ ($+1$)};
	\node[o,above left =of ch](g) {$O_2$};

	\path
	(s) edge node[el]{$-2^N$} (a)
	(a) edge node[el]{$-N$} (b)
	(b) edge[bend right] node[el]{$-1$} (c)
	(c) edge node[el]{$-1$} (d)
	(d) edge node[el]{$-1$} (e)
	(e) edge node[el]{$-1$} (ch)
	(ch) edge[bend right] node[el]{$+1$} (g)
	(g) edge node[el]{$0$} (a)
	(g) edge[bend right,densely dotted] node[el]{$-1$} (s)
	(g) edge[densely dotted] node[el]{$0$} (s)
	;
	\end{tikzpicture}
	\end{center}
	\caption{Overall structure of the game for succinct Hamiltonian cycle}
	\label{fig:overallGad}
	\end{figure}
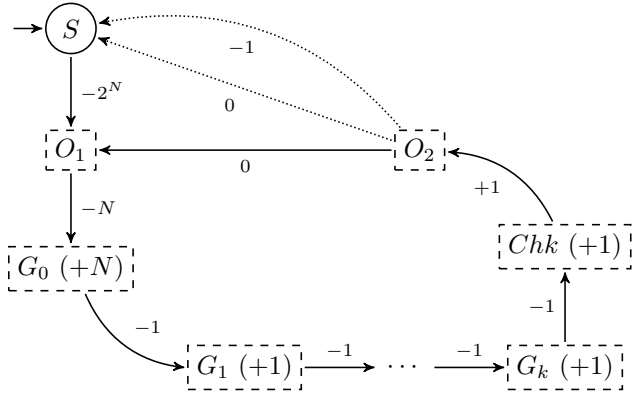

	\item \paragraph*{Simplifying assumptions}
	Let us assume inputs of the circuit $C$ are labelled $x_1, \dots x_{2N}$
	and that it has $k$ gates $G_1, \dots, G_k$ numbered in an order that
	respects the circuit graph, so $G_j$ has inputs from $\{x_i, \neg x_i :
	1 \leq i < 2N+j\}$ where, for convenience, $x_{2N+i}$ indicates the
	output of gate $G_i$. We may assume each gate has two inputs and (as we
	are allowing negated inputs) we may assume we only have AND and OR
	gates.

	\item \paragraph*{Construction description}
	The game consists of two \emph{external} actions primarily for transitions
	between observations: $\sigma$ (solid lines) and $\sigma'$ (dotted
	lines); and a number of \emph{internal} actions denoted with $\tau$ and $\chi$ for
	transitions primarily within observations (not shown). The numbers in
	parentheses indicate the maximum weight that can be added to the total with
	internal transitions, and the edge weights indicate the weight of \emph{all}
	transitions between observations.

	Our game proceeds in several stages:
	\begin{enumerate}[nolistsep]
		\item The transition from $S$ to $O_1$ sets the initial (succinct)
			vertex (stored in a subset of the states of $O_1$) and
			initializes the vertex counter to $-2^N$.
		\item Internal transitions in $G_0$ select the next vertex, the
			transition from $O_1$ to $G_0$ initializes this procedure.
		\item For $i>0$, internal transitions in $G_i$ evaluate gate $i$, incoming
			transitions initialize this by passing on the previous evaluations
			(including the current and next vertices).
		\item Internal transitions in $Chk$ test if the circuit evaluates to $1$.
		\item The next succinct vertex (chosen in $G_0$) is passed to $O_2$,
			where there is an implicit check that this vertex has not been
			visited before, and the counter is incremented.
		\item The play can return to $S$, generating a mixed lasso if and only
			if the vertex counter is $0$, \ie $2^N$ vertices have been
			correctly visited, or return to $O_1$ with a new current
			succinct vertex.
	\end{enumerate}
	The weights on the incoming transitions to an observation are designed to impose
	a penalty that can only be nullified if the correct sequence of internal
	transitions is taken. We observe that if there is a penalty that is not
	nullified then the game can never enter a mixed lasso (as the vertex counter
	will still be negative when a vertex is necessarily revisited).
	The overall (\ie observation-level) structure of the game is
	shown in Figure~\ref{fig:overallGad}.
	
	We now describe the structure of the observations.

	\item \paragraph*{Observation $O_1$}
	It contains $2N$ states: $\{x_i, \overline{x}_i \st 1 \leq i \leq N \}$.
	For convenience we will use the same labels across different observations, using
	observation membership to distinguish them. There are $\sigma$-transitions from
	$S$ to $\{x_i \st 1 \leq i \leq N\}$ with weight $-2^N$.

	\item \paragraph*{Observation $O_2$} It contains $2N+1$ states:
	$\{x_i, \overline{x}_i \st 1 \leq i \leq N \} \cup \{\bot\}$.
	There are $\sigma$-transitions from each state in $O_2$ other
	than $\bot$ to its corresponding state in $O_1$ with weight $0$.
	There is a $\sigma'$-transition from each state in $O_2$ other
	than $\bot$ to $S$ with weight $0$, and a $\sigma'$-transition
	from $\bot$ to $S$ with weight $-1$.

	\item \paragraph*{Observation $G_0$}
	It contains $5N$ states: $\{x_i, \overline{x}_i \st 1 \leq i \leq 2N \}
	\cup \{y_i \st N < i \leq 2N \} $. There is a $\sigma$-transition from each
	state in $O_1$ to its corresponding state in $G_0$ of weight $-N$ and in
	addition, $\sigma$-transitions from every state in $O_1$ to $\{y_i \st N < i
	\leq 2N \} $ also of weight $-N$. For $N < j \leq 2N$ there is a $\tau_j^+$
	transition of weight $1$ from $y_j$ to $x_j$ and a $\tau_j^-$ transition of
	weight $1$ from $y_j$ to $\overline{x}_j$. For all states in $G_0$ other than
	$y_j$ there is a $\tau_j^+$ and $\tau_j^-$ loop of weight $1$. Figure~\ref{fig:g0gad}
	shows the construction.

	\begin{figure}
	\begin{center}
	\begin{tikzpicture}[v/.style={draw,circle,minimum size=2pt,inner
		sep=2pt},semithick]
	\coordinate (top);
	\coordinate (left) at ($(top)+(-60pt,-50pt)$);
	\coordinate (right) at ($(top)+(60pt,-50pt)$);

	\foreach \s/\t in {-30pt/1, -15pt/2, 30pt/N}{
		\node at ($(top)+(\s,-15pt)$) [v,label=above:$x_{\t}$] {};
		\node at ($(top)+(\s,-27pt)$) [v,label=below:$\overline{x}_{\t}$] {};
	}
	\node at ($(top)+(5pt,-21pt)$) {$\cdots$};
	\node[fit = (top) (left) (right), draw=none](id){};
	\draw[-,dotted] (left) -- (right);

	\node[below left =of id, xshift=-2cm, draw, dashed](o1){$O_1$};

	\foreach \s/\t in {-30pt/N+1, -120pt/2N}{ %80,160
		\node (y) at ($(left)+(40pt,\s)$) [v,label=below left:$y_{\t}$] {};
		\node (x) at ($(right)-(40pt,-10pt-\s)$) [v,label=right:$x_{\t}$] {};
		\node (x') at ($(right)-(40pt,10pt-\s)$)
		[v,label=right:$\overline{x}_{\t}$] {};
		\path (y) edge[bend left] node[el,pos=0.5]{$\tau^+_{\t}$} (x)
			(y) edge[bend right] node[el,pos=0.7,swap](bottom){$\tau^-_{\t}$}(x')
			(o1) edge[densely dotted] node[el]{$\sigma$} node[el,swap]{$-N$} (y)
		;
	}

	\node at ($0.5*(left)+0.5*(right)+(0pt,-75pt)$) {$\vdots$};

	\node[fit= (top) (left) (right) (bottom),dashed,label=right:$G_0$](g0){};
	\path (o1) edge[loosely dotted,bend left] node[el]{$\sigma$} node[el,swap]{$-N$}
	(id);
	\end{tikzpicture}
	\end{center}
	\caption{Gadget for $G_0$}
	\label{fig:g0gad}
	\end{figure}
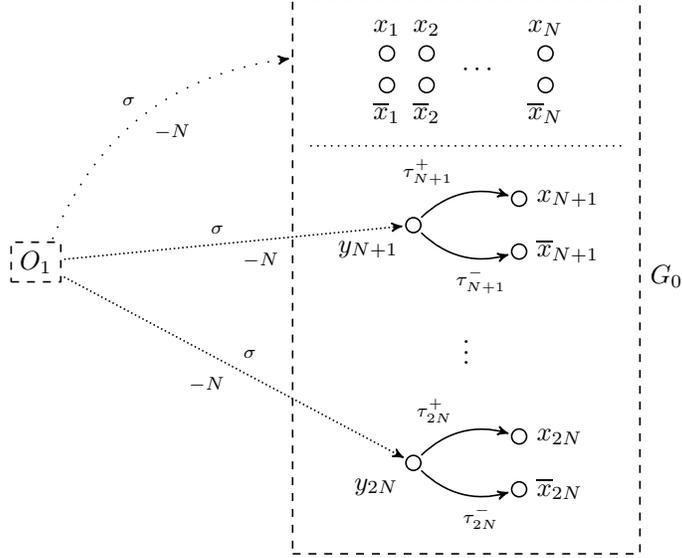

	\item \paragraph*{Observations $G_j$ (for $j>0$)} Observation gadgets
	for the logical gates follow the idea laid out earlier.
	The observation corresponding to gate $j$
	contains $4N+2j+8$ states: $\{x_i, \overline{x}_i \st 1 \leq i \leq 2N+j \} \cup
	\{v_m, \overline{v}_m\st0 \leq m \leq 3\}$. Recall gate $j$ has inputs from
	$\{x_i, \overline{x}_i\st1 \leq i < 2N+j\}$.  Suppose these inputs are $y_l \in
	\{x_l,\overline{x}_l\}$ and $y_r \in \{x_r,\overline{x}_r\}$, and for
	convenience let $\overline{y}_l$ and $\overline{y}_r$ denote the other member of
	the pair (\ie the complement of the input). We have a $\sigma$-transition of
	weight $-1$ from $\{x_i, \overline{x}_i \st 1 \leq i < 2N+j \} \subseteq
	G_{j-1}$ to the corresponding vertex in $G_j$. In addition we have
	$\sigma$-transitions of weight $-1$ from $y_l, \overline{y}_l, y_r,
	\overline{y}_r \in G_{j-1}$ to $v_0, \overline{v}_0, v_1, \overline{v}_1 \in
	G_j$ respectively. We have the following internal transitions:
	\begin{itemize}[nolistsep]
		\item $\tau^+$ (weight $0$): $v_1$ to $v_2$, $\overline{v}_1$ to
			$\overline{v}_3$, $v_0$ to $x_{2N+j}$, $\overline{v}_0$ to
			$\overline{x}_{2N+j}$ if gate $j$ is an AND gate,
			$\overline{v}_0$ to $x_{2N+j}$ if it is an OR gate,
		\item $\tau^-$ (weight $0$): $v_1$ to $v_3$, $\overline{v}_1$ to
			$\overline{v}_2$, $\overline{v}_0$ to $\overline{x}_{2N+j}$,
			$v_0$ to $\overline{x}_{2N+j}$ if gate $j$ is an AND gate, $v_0$
			to $x_{2N+j}$ if it is an OR gate,
		\item $\chi$ (weight $1$): $v_2$ to $v_3$, $\overline{v}_2$ to
			$\overline{v}_3$.
	\end{itemize}
	For all states in $G_j$ we have $\tau^\pm,\chi$-loops with the same weights
    as above (\ie
	$\chi$ loops have weight $1$, $\tau^{\pm}$ loops have weight $0$).
    %As in
	%our sample gate before, we add from $v_3$ and $\overline{v}_3$ external
	%transitions to a state $\bot_0$ which is winning for Eve. Additionally, they
    %transition to the $y_r$ and $\overline{y}_r$ states, respectively, of the
    %next observation. (So that punished
    %cheats carry over.) 
    Bit states (\ie~$x_{N+1}$, $\overline{x}_{2N}$)
    transition to the next observation on external actions.

	Figure~\ref{fig:gjgad} shows an example of the construction of $G_j$ for the
	gate $x_l \wedge \neg x_r$.

	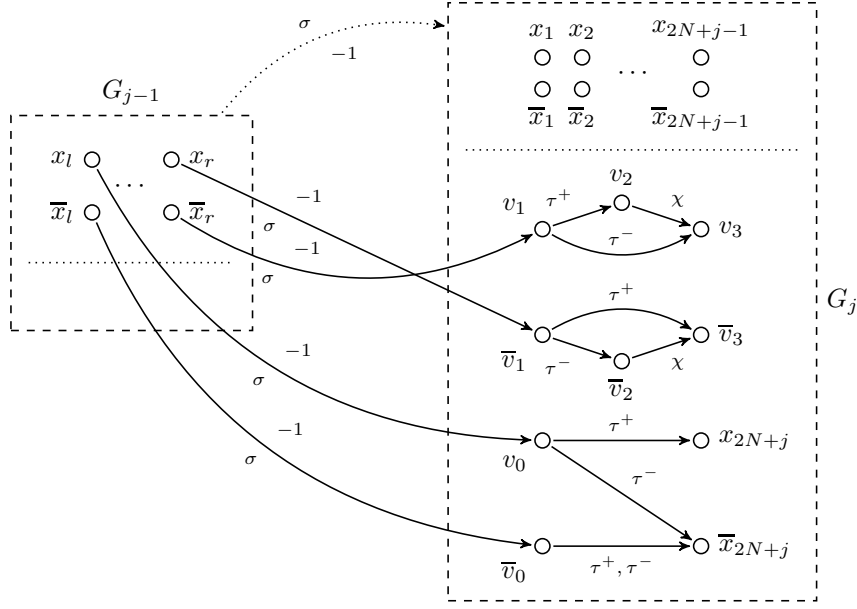
\begin{figure}
	\begin{center}
	\begin{tikzpicture}[v/.style={draw,circle,minimum size=2pt,inner
		sep=2pt},semithick]
	\coordinate (top);
	\coordinate (left) at ($(top)+(-60pt,-50pt)$);
	\coordinate (right) at ($(top)+(68pt,-50pt)$);

	\foreach \s/\t in {-30pt/1, -15pt/2, 30pt/2N+j-1}{
		\node at ($(top)+(\s,-15pt)$) [v,label=above:$x_{\t}$] {};
		\node at ($(top)+(\s,-27pt)$) [v,label=below:$\overline{x}_{\t}$] {};
	}
	\node at ($(top)+(5pt,-21pt)$) {$\cdots$};
	\node[fit = (top) (left) (right), draw=none](id){};
	\draw[-,dotted] (left) -- (right);

	\node at ($(left)+(30pt,-30pt)$) [v,label =above left:$v_1$] (v1) {};
	\node at ($(v1)+(30pt,10pt)$) [v,label=above:$v_2$] (v2) {};
	\node at ($(v2)+(30pt,-10pt)$) [v,label=right:$v_3$] (v3) {};

	\node at ($(v1)+(0,-40pt)$) [v,label =below left:$\overline{v}_1$] (v1') {};
	\node at ($(v1')+(30pt,-10pt)$) [v,label=below:$\overline{v}_2$] (v2') {};
	\node at ($(v2')+(30pt,10pt)$) [v,label=right:$\overline{v}_3$] (v3') {};

	\path
		(v1) edge[bend right] node[el]{$\tau^-$} (v3)
		(v1) edge node[el]{$\tau^+$} (v2)
		(v2) edge node[el]{$\chi$} (v3)
		(v1') edge[bend left] node[el]{$\tau^+$} (v3')
		(v1') edge node[el,swap]{$\tau^-$} (v2')
		(v2') edge node[el,swap]{$\chi$} (v3');
		
	\node at ($(v1')+(0,-40pt)$) [v,label =below left:$v_0$] (v0) {};
	\node at ($(v0)+(0,-40pt)$) [v,label =below left:$\overline{v}_0$](v0'){};
	\node at ($(v0)+(60pt,0)$) [v,label =right:$x_{2N+j}$] (x){};
	\node at ($(v0')+(60pt,0)$) [v,label =right:$\overline{x}_{2N+j}$] (x'){};	

	\path
	(v0) edge node[el]{$\tau^+$} (x)
	(v0) edge node[el]{$\tau^-$} (x')
	(v0') edge node[swap,el](bottom){$\tau^+, \tau^-$} (x');

	\node[fit = (top) (left) (right) (bottom), dashed, label=right:$G_j$](gj){};

	\coordinate (gr) at ($(left)-(3cm,1.5cm)$);
	\coordinate (gl) at ($(gr)-(80pt,0)$);
	\coordinate (gt) at ($(gr)-(40pt,-50pt)$);
	\coordinate (gb) at ($(gr)-(40pt,20pt)$);
	\draw[-,dotted] (gl) -- (gr);

	\node[fit = (gt) (gl) (gr) (gb), dashed, label=above:$G_{j-1}$]{};
	\node[fit = (gt) (gl) (gr), draw=none](id'){};

	\path
		(gt) ++(-30pt, -21pt) node {}
		++(15pt,10pt) node[v, label=left:$x_l$](xl) {}
		++(0,-20pt) node[v, label=left:$\overline{x}_l$](xl') {}
		++(15pt,10pt) node {$\cdots$}
		++(15pt,10pt) node[v, label=right:$x_r$](xr) {}
		++(0,-20pt) node[v,label=right:$\overline{x}_r$](xr') {}
		++(15pt,10pt) node {};

	\path
		(xl) edge[bend right] node[el,swap]{$\sigma$} node[el]{$-1$} (v0)
		(xl') edge[bend right] node[el,swap]{$\sigma$} node[el,]{$-1$} (v0')
		(xr) edge node[el,pos=0.3,swap]{$\sigma$} node[el,pos=0.3]{$-1$} (v1')
		(xr') edge[bend right] node[el,pos=0.3,swap]{$\sigma$}
		node[el,pos=0.3]{$-1$} (v1)
	;
	\path (id') edge[dotted,bend left] node[el]{$\sigma$} node[el,swap]{$-1$} (id);

	\end{tikzpicture}
	\end{center}
	\caption{Gadget for gate $x_l \wedge \neg x_r$ (self-loops not shown)}
	\label{fig:gjgad}
	\end{figure}

	\item \paragraph*{Observation $Chk$} This last observation gadget
	contains $4N+2$ states: $\{x_i, \overline{x}_i \st 1 \leq i \leq
	2N \} \cup \{y,z\}$. There is a $\sigma$-transition of weight $-1$ from
	$\{x_i, \overline{x}_i \st 1 \leq i \leq 2N \}\subseteq G_k$ to their
	corresponding states in $Chk$, and a $\sigma$-transition of weight $-1$
	from $x_{2N+k} \in G_k$ to $y$. There is a $\chi$-transition of weight
	$1$ from $y$ to $z$ and for all other states in $Chk$ there is a
	$\chi$-loop of weight $1$.  There is a $\sigma$-transition of weight $1$
	from all states in $Chk$ to $\bot \in O_2$ and for $N < i \leq 2N$ there
	is a $\sigma$-transition of weight $1$ from $x_i \in Chk$ to $x_{i-N}
	\in O_2$ and from $\overline{x}_i \in Chk$ to $\overline{x}_{i-N} \in
	O_2$.

	\item \paragraph*{Correctness of the construction}
	We present a similar argument to that given for the proof of
	Theorem~\ref{thm:isadeqpure}. Recall the game's initial transition
	is weighted $-2^N$. Further, note that internal transitions in all
	observations can only lead to reaching a good or bad sink or reaching
	the next observation gadget (while nullifying the incoming $-1$ weight).
	Hence, completing $2^N$ full simulations of the circuit is the only way
	of not forming a bad cycle and reaching observation $O_2$ with all
	concrete paths having weight $0$. From there, a mixed cycle can be
	formed by going back to $S$. The latter thus holds if and only if the
	graph encoded succinctly by the given circuit has a Hamiltonian cycle.
\end{proof}

Based on the construction used to prove the above result, we will now show
hardness for forcibly FBC class membership.

\begin{lemma}
	Let $G$ be an MPG with partial observation. Determining if $G$ is
	forcibly FBC is
	\NEXP-hard.
\end{lemma}
\begin{proof}
	Suppose we make the following adjustments to the construction given in the proof
	of Lemma~\ref{lem:hardness-fbc}:
	\begin{itemize}[nolistsep]
		\item Change the weights of incoming transitions to $G_i$
			($i>0$) to $-5$ and the weights of all internal
			$\tau$-transitions to $1$,
		\item Change the weight of the $\sigma'$-transition from $\bot
			\in O_2$ to $S$ to $0$,
		\item Add a new state $\bot$ to all observations other than $S$
			(and $O_2$),
		\item Add a $\sigma$-transition of weight $2^N$ from $S$ to
			$\bot \in O_1$, and
		\item Whenever there is a transition of weight $w$ from $x_i \in
			o$ to $x_j \in o'$ ($o,o'$ and $i,j$ possibly the same)
			add a transition of weight $-w$ from $\bot \in o$ to
			$\bot \in o'$.
	\end{itemize}
	Then the only possible non-mixed lasso in the resulting graph\footnote{We assume
	dead-ends go to a dummy state with a single mixed self-loop.} is one that would
	correspond to a successful traversal of a Hamiltonian cycle. Eve can force the
	play to this cycle if and only if the succinct graph has a Hamiltonian cycle.
\end{proof}

Somewhat surprisingly, for the winner determination problem we have an \EXP
algorithm matching the \EXP-hardness lower bound from games with visible weights.
This is in contrast to the class membership problem in which an exponential
increase in complexity occurs when moving from limited to partial observation.
\begin{theorem}[Winner determination]
	\label{thm:pureWin}
	Let $G$ be a forcibly FBC MPG.
	Determining if Eve wins $G$ is \EXP-complete.
\end{theorem}
\begin{proof}
The lower bound follows from the fact that forcibly FBC games are a
generalization of visible weights games (see Lemma~\ref{lem:pure-visible}),
shown to be \EXP-complete in~\cite{ddgrt10}. For the upper bound, rather than
working on the reachability game $\Gamma'$ associated to $G'$, which is doubly-exponential
in the size of $G$, we instead reduce the problem of determining the winner to
that of solving a safety game which is only exponential in the size of $G$.

Given an MPG with partial observation $G = ( Q, q_I,\Sigma,\Delta,w,\Obs
)$, let $G' = ( Q',
q'_I,\Sigma,\Delta',w',\Obs' )$ be its limited-observation version and
$\Gamma'$ be the finite reachability game, as defined in
Section~\ref{sec:adeqpure-inc}, constructed for $G'$ (not for $G$!). Let
$\mathcal{E} = [-1, 2 W |\Obs'|] \cup \{ \bot \}$ where $W =
\max\{|w(e)| : e \in \Delta\}$, and let $\mathcal{B}' \subseteq \mathcal{B}$
be the set of functions $f : Q \to \mathcal{E}$.

The safety game will be played on $\mathcal{B}'$ with the transitions defined by
$\sigma$-successors. The idea is that a given position $f \in \mathcal{B}'$ of
the safety game corresponds to being in an observation of $G'$, namely
$\supp(f)$. The functions additionally keep track of the minimal weight of all
concrete paths ending in states from $\supp(f)$. However, they do so only up to
the point where a belief cycle is formed. Since $W$ is the biggest absolute
weight in $G$ and in $G'$, and the length of any simple belief path is bounded
by $|\Obs'|$, it suffices to keep track of weights from $\mathcal{E}$.

Formally, the safety game is $\mathcal{S}_G = ( \mathcal{B}',
f'_I, \Sigma, \Delta_{succ}, \mathcal{F}'_{neg} )$ where $f'_I(q_I) = W
|\Obs'|$ and $f'_I(q) = \bot$ for all other $q \in Q$; $(f,\sigma,f') \in
\Delta_{succ}$ if $f'$ is a proper $\sigma$-successor of $f$ where we let
\[
	a + b =
	\begin{cases}
		\bot&\text{if $a=\bot$ or $b=\bot$},\\
		-1 &\text{if $a=-1$, $b=-1$, or $a+b<0$, and}\\
		\min\{a+b, 2 W |Q|\}&\text{otherwise.}
	\end{cases}
\]
$\mathcal{F}'_{neg}$ is the
set of all functions $f \in \mathcal{F}'$ such that $f(q) = -1$ for some $q \in
\supp(f)$. The game is played similar to the reachability game $\Gamma$, \ie
Eve chooses an action $\sigma$ and Adam resolves non-determinism by selecting a
proper $\sigma$-successor. In this case, however, Eve's goal is to avoid
visiting any function in $\mathcal{F}'_{neg}$.

In this safety game (just like in the weighted unfolding) the non-negative
integer values of $f$ give a lower bound for the minimum weights of the concrete
paths ending in the given state (see Lemma~\ref{lem:minPath}). More formally,
since obtaining a $-1$ weight means that henceforth the weight stays $-1$, we
have that if $f(q) \neq \bot$ and $f(q) \geq 0$ then the minimum weight over all
concrete paths starting at $q_I$ and ending at $q$ is at least $f(q) + W
|\Obs'|$. We do not have equality because of the max applied after each sum. If
$f(q) = -1$ then there is a concrete path of weight at most $ - W |\Obs'| - 1$,
because $f_I(q_I) = W|\Obs'|$.  As the winner of a forcibly FAC game can be
resolved in at most $|\Obs'|$ transitions it turns out that this is sufficient
information to determine the winner.

The above observation that non-negative values give lower bounds for concrete
paths ending at the given state implies that if Eve has a strategy to always
avoid $\mathcal{F}'_{neg}$ then $\liminf_{n \to \infty} \frac{\pi[..n]}{n} \geq
0$ for all concrete paths $\pi$ consistent with the play. That is, if Eve has a
winning strategy in $\mathcal{S}_G$ then she has a winning strategy in $G$.

Now suppose Eve has a winning strategy in $G$. It follows from the determinacy
of forcibly FAC games and Theorem~\ref{thm:stratTransfer} that she has a
winning strategy $\lambda$ in $\Gamma'$.  Let $\lambda^*$ be the
translation of $\lambda$ to $G'$ as per Theorem~\ref{thm:stratTransfer}, and let
$M$ denote the set of memory states required for $\lambda^*$.  Clearly
$\lambda^*$ induces a strategy in $\mathcal{S}_G$. We claim this induced
strategy is winning in $\mathcal{S}_G$. Let $\rho = f_0 \sigma_0 \dots$ be any
play in $\mathcal{S}_G$ consistent with $\lambda^*$, and let $\mu_i = g_i^{(0)}
\dots g_i^{(n_i)}$ denote the
$i$-th memory state obtained in the generation of $\rho$ (as in
Lemma~\ref{lem:alwaysGood}).  Then, with a slight adjustment to the proof of
Lemma~\ref{lem:alwaysGood} to account for function values not exceeding $2 W
|\Obs'|$ we have for all $i$ and all $q$:
\begin{align*}
	f_i(q) - W |\Obs'| & \geq g_i^{(n_i)}(q)\\
	& = \min \{w(\pi) \st \pi \in \gamma(\supp(\mu_i)) \text{ and $\pi$ ends at
		$q$}\} \footnotemark\\
	& \geq - W |\Obs'|
\end{align*}
\footnotetext{The step follows from Lemma~\ref{lem:minPath}.}
because $|\mu_i| \leq |\Obs'|$ from the definition of $\Gamma'$.  Thus
$f_i(q) \geq 0$ for all $i$. Hence $\rho$ does not reach $\mathcal{F}'_{neg}$
and is winning for Eve.  Thus $\lambda^*$ is a winning strategy for Eve.

So to determine the winner of $G$, it suffices to determine the winner of
$\mathcal{S}_G$.  This is just the complement of alternating reachability, known
to be decidable in polynomial time (see \eg~\cite{pap03}).  As
\[
	|\mathcal{S}_G| = O(|\mathcal{F}'|^2) = O\left((2 W
		|\Obs'|+1)^{|Q|}\right) = 2^{O(|Q|^2)},
\]
determining the winner of $\mathcal{S}_G$, and hence $G$, is in \EXP.
\end{proof}

\begin{corollary}
	Let $G$ be an FBC MPG. Determining if Eve wins $G$ is \EXP-complete.
\end{corollary}

\bibliographystyle{abbrv}
\bibliography{mpay-ii}

\end{document}